\documentclass[a4paper,11pt]{article}

\usepackage{a4wide}
\usepackage[english]{babel}
\usepackage{amssymb,amsmath,amscd,amsfonts,amsthm,bbm,color}
\usepackage{hyperref}
\usepackage{graphicx}
\usepackage{tikz,pgfplots}

\renewcommand{\cal}{\mathcal}

\newcommand{\toaslong}{\xrightarrow[n\to\infty]{\text{a.s.}}}

\DeclareMathOperator{\var}{\mathbb Var}

\DeclareMathOperator{\dom}{{\cal D}}
\DeclareMathOperator{\ran}{{\cal R}}
\DeclareMathOperator{\tr}{Tr}

\DeclareMathOperator{\support}{supp}
\DeclareMathOperator{\rank}{rank}

\newcommand{\1}{\mathbbm 1}

\newcommand{\NN}{\mathbb N}
\newcommand{\ZZ}{\mathbb Z}

\newcommand{\RR}{\mathbb R}
\newcommand{\CC}{\mathbb C}
\newcommand{\PP}{\mathbb P}
\newcommand{\EE}{\mathbb E}

\newcommand{\bs}{\boldsymbol}

\theoremstyle{plain} 

\newtheorem{theorem}{Theorem}[section]
\newtheorem{lemma}{Lemma}[section]

\newtheorem{proposition}{Proposition}[section]

\newcommand{\tQ}{{\widetilde Q}}
\newcommand{\tS}{{\widetilde S}}
\newcommand{\td}{{\tilde\varphi}}

\newcommand{\bd}{{\bs\varphi}}
\newcommand{\tbd}{{\tilde{\bs\varphi}}}

\newcommand{\bz}{{\bs\zeta}}
\newcommand{\tbz}{{\tilde{\bs\zeta}}}

\renewcommand{\SS}{{\bs S}}
\newcommand{\tSS}{\widetilde{\bs S}}

\newcommand{\lf}{{\|}}
\newcommand{\rf}{\|_{\text{fro}}}

\begin{document}

\title
{The Shannon's mutual information of a multiple antenna time 
and frequency dependent channel: an ergodic operator approach}

\author{W. Hachem\thanks{CNRS LTCI; Telecom ParisTech, France 
(\texttt{walid.hachem@telecom-paristech.fr}),}, 
A. Moustakas\thanks{National \& Capodistrian University of Athens, 
Department of Physics, Greece (\texttt{arislm@phys.uoa.gr}),},  
and L. Pastur\thanks{Institute of Low Temperature Physics, Kharkiv, Ukraine
(\texttt{pastur@ilt.kharkov.ua}).}}

\date{\today}
\maketitle

\begin{abstract} 
Consider a random non-centered multiple antenna radio transmission channel.
Assume that the deterministic part of the channel is itself frequency
selective, and that the random multipath part is represented by an ergodic
stationary vector process.  In the Hilbert space $l^2(\ZZ)$, one can associate
to this channel a random ergodic self-adjoint operator having a so-called
Integrated Density of States (IDS). Shannon's mutual information per receive
antenna of this channel coincides then with the integral of a $\log$ function
with respect to the IDS.  In this paper, it is shown that when the numbers of
antennas at the transmitter and at the receiver tend to infinity at the same
rate, the mutual information per receive antenna tends to a quantity that can
be identified and, in fact, is closely related to that obtained within the
random matrix approach~\cite{tel-95}.  This result can be obtained by analyzing
the behavior of the Stieltjes transform of the IDS in the regime of the large
numbers of antennas.
\end{abstract}

\section{Introduction and problem statement}
\label{intro}

In the landmark papers by Foschini and Gans \cite{Foschini1998_BLAST1} and by
Telatar \cite{tel-95} the great promise of the use of multiple
transmit and receive antennas was presented and established. The importance of
such Multiple Input Multiple Output (MIMO) links is based on the fact that
parallel data streams emanating from different transmit antennas can be decoded
simultaneously from the receive array, thus making the throughput scale
linearly with the number of transmit antennas. This scaling property is
important in the quest to meet the expected thousand-fold increase of wireless
network capacity in the coming years\cite{Qualcomm} . At the same time, the
necessary transmit power per data stream is reduced by the same factor, thus
also addressing energy-related pollution issues, which are becoming major
societal and economical concerns \cite{EARTH_D23,
Chen2011_FundamentalTradeoffs_Green}. In fact, so-called {\em massive} MIMO
systems of a few hundred antenna arrays \cite{Larsson2014, Rusek2013} have been
identified as a key enabling technology for the next generation 5G wireless
networks \cite{Andrews2014_WhatWill5GBe}, promising unprecedented data transfer
increases.

In this paper, we address the calculation of Shannon's mutual information of
MIMO channels under some general assumptions on the channel statistics who have
been considered only partially in the literature. To start with, we shall 
assume that
the channel is subjected to time correlations that are more realistic than the
commonly assumed multi-block-fading
model~\cite{Caire1999_PowerControlFadingChannels, Biglieri1998_FadingChannels}.
A second assumption on the model is related to the so-called frequency
selectivity of the channel due to the delay spread induced by the
reverberations of the signal from buildings, walls, etc.  Typically, this
effect is modelled using a tap-delay line system with different tap
variances~\cite{Biglieri1998_FadingChannels, Biglieri2004_MIMO_review}. A third
component of the channel model is related with the so-called Ricean (or
``deterministic'') part, due to \emph{e.g.}~a line of sight or a specular
component. In this work, we consider in most generality that the deterministic
part of the channel is also frequency selective~\cite{almers2007survey}.


To be more specific, let us denote by $T$ the number of antennas at the 
transmitter and by $N$ the number of antennas at the receiver.
Then the $\CC^N$-valued signal received at time $k$ according to our model is
\[
Y(k) = \sum_{\ell=-L}^L H(k,\ell) S(\ell) + V(k)
\]
where the processes $\{S(k)\}_{k \in\ZZ}$ and $\{V(k)\}_{k \in\ZZ}$ represent
respectively the $\CC^T$-valued input signal fed to the channel and the
$\CC^N$-valued additive noise.
The channel with $2L+1$ matrix coefficients is represented
$\CC^{N\times (2L+1)T}$-valued random process
$\{{\bs H}(k) = [ H(k,k-L), \ldots, H(k,k+L) ]\}_{k\in\ZZ}$, assumed to be
Gaussian, stationary, ergodic, and generally non-centered.
It is assumed that the processes $\{S(k)\}$, $\{V(k)\}$ and
$\{{\bs H}(k)\}$ are mutually independent, and that $\{S(k)\}$
(respectively $\{V(k)\}$) is an independent process such that
$S(k) \sim \cal{CN}(0, I_{T})$ (resp.~$V(k) \sim \cal{CN}(0, I_{N}))$ for
any $k\in\ZZ$.

As is well-known, the Doppler effect induced by the mobility of the
communicating terminals determines the form of the covariance function of the
process $\{{\bs H}(k)\}$.
The multipath effect that induce the frequency selectivity is captured by 
the $2L+1$ matrix channel coefficients
which are subjected in the practical situations to a certain power profile.
In addition, the fact that $\EE {\bs H}(0)$ can be non zero and can be also  
frequency selective is due \emph{e.g.} a line of sight or a specular component. 
We shall formulate our assumptions on the channel more precisely in 
Section~\ref{results} below.  \\

The signal $Y^n = [Y(-n)^T, \ldots, Y(n)^T]^T$ observed
during the time window $(-n,\ldots,n)$ satisfies the equation
$Y^n = H^n S^n + V^n$ where $V^n = [V(-n)^T, \ldots, V(n)^T]^T$,
$S^n = [S(-n-L)^T, \ldots, S(n+L)]^T$, and
\begin{equation}
\label{HnHn}
H^n = \begin{bmatrix}
H(-n,-n-L) & \cdots & H(-n,-n+L) & & 0\\
            & \ddots &        & \ddots \\
   0        &        & H(n,n-L) & \cdots & H(n,n+L)
\end{bmatrix}
\end{equation}
is a $(2n+1)N \times (2n+2L+1)T$ matrix.
Our goal is to study the mutual information of this channel assuming that it is
perfectly known at the receiver.  More precisely, we study $I(S; (Y,H))$, the
mutual information between $\{ S(k) \}$ and the couple 
$( \{ Y(k) \}, \{ {\bs H}(k) \})$. It is given by the equation
\[
I(S; (Y,H)) = \limsup_n \frac{1}{2n+1} I( S^n; (Y^n, H^n))
\]
where 
\[
I( S^n; (Y^n, H^n)) = \EE \log\det( H^n H^{n*} + I) 
\]
is the mutual information between $S^n$ and $(Y^n, H^n)$ 
(see~\cite[Chap. 8]{gray-entropy11}, \cite{tel-95} and 
Section~\ref{ergo-op} below). \\ 

We shall see that the natural way to tackle this problem is to consider 
the MIMO channel as a \emph{random operator} represented
by the doubly infinite matrix $H = [ H(k,\ell) ]_{k,\ell\in\ZZ}$ and acting
on the Hilbert space $l^2(\ZZ)$ of the doubly infinite square summable 
sequences. Due to the ergodicity of $({\bs H}(k))$, this operator
that we also denote as $H$ is \emph{ergodic} in the sense of \cite{pas-fig}.
By the ergodicity, it turns out that the self-adjoint operator $HH^*$ where
$H^*$ is the adjoint of $H$ has a so-called \emph{Integrated Density of
States} (IDS). Specifically, there exists a positive deterministic measure
$\mu$ of total mass one such that on a probability one set,
\[
\frac{1}{(2n+1)N} \tr g(H^n H^{n*}) \xrightarrow[n\to\infty]{}
\int g(\lambda) \, \mu(d\lambda)
\]
for any continuous and bounded real function $g$.
The IDS is the distribution
function of $\mu$. The crucial observation is that this convergence leads to
the convergence of $I( S^n; (Y^n, H^n)) / (2n+1)$, and the limit is
\[
I(S; (Y,H)) = N \int \log(1+\lambda) \, \mu(d\lambda)
\]
as it will be shown below. Our goal is then to study the behavior of the
integral at the right hand side. Unfortunately, a few can be said about this
behavior in the general situation. To circumvent this problem, one needs to
resort to a certain asymptotic regime.

In this paper, consistently with an established practice in the evaluation
of the mutual information of MIMO channels, we consider an asymptotic regime 
where the numbers of antennas $N$ and $T$ tend to infinity at the same 
rate\footnote{We note that other asymptotic regimes can also be studied by
the techniques of this paper, for instance
the one where $N$ and $T$ are fixed and where $L$ tends to infinity.}.
A central object of study will be the Stieltjes transform of the 
measure $\mu$, that is the complex analytical function
\[
\bs m(z) = \int \frac{1}{\lambda - z} \mu(d\lambda)
\]
on the upper half-plane $\CC_+ = \{ z \,: \, \Im z > 0 \}$. This function
completely characterizes $\mu$ and is intimately connected with the resolvent
\[
Q(z) = (HH^* - z I)^{-1}
\]
of the operator $HH^*$. Specifically, considering the matrix block
representation $Q(z) = [ Q(k,\ell)(z) ]_{k,\ell\in\ZZ}$ of the resolvent where
the blocks $Q(k,\ell)(z)$ are $N\times N$ complex matrices, it holds that
\[
\bs m(z) = \frac{\tr\EE Q(0,0)(z)}{N} .
\]
The core of our analysis consists in studying the behavior of the right hand
side of this relation in the large $N,T$ regime.
Re-denoting $\mu$ and $\bs m(z)$ as $\mu_T$ and $\bs m_T(z)$ to stress the
dependency on $T$, it turns out that there exists a sequence of probability
measures $\bs \pi_T$ which approximates $\mu_T$ in the sense that
\[
\bs m_T(z) - \int \frac{1}{\lambda - z} \bs\pi_T(d\lambda)
\xrightarrow[T\to\infty]{} 0 .
\]
In addition, we can associate to each $\bs\pi_T$ a positive
number $\cal I_T$ of order one for large $T$, and such that
\[
\frac 1N I(S; (Y,H)) - \cal I_T \xrightarrow[T\to\infty]{} 0.
\]
Thus, the mutual information increases at a linear scale with the number
of antennas at one side of the transmission. 
The Stieltjes transform of $\bs\pi_T$ as well as $\cal I_T$ satisfy
a system of equations that can be easily implemented on a computer.
Solving for $\cal I_T$, one can evaluate the impact of the various statistical
parameters of the channel model on the mutual information. \\

In the literature dealing with the ergodic mutual information of MIMO channels,
it has been frequently assumed that the channel was a
\emph{frequency non selective}
channel which satisfies the simplifying \emph{block-fading} assumption
(see \emph{e.g.}~\cite{tel-95}).
In the language of the ergodic operator theory, the operator $HH^*$
corresponding
to such channels is a block diagonal operator with independent and identically
distributed diagonal blocks, and the existence of the IDS is immediate. In
order to
approximate the Stieltjes transform of $\mu_T$ in the regime of the large
number of antennas, one only needs to study the behavior of
\[
\frac 1N \tr \EE Q(0,0)(z) = \frac 1N \tr \EE (H(0,0) H(0,0)^* - z I_N)^{-1}.
\]
The large $N,T$ behavior of this object can in fact be studied with the help of
random matrix theory techniques, without making any explicit reference to
ergodic operators.  In this framework, a
number of works studied $N^{-1} \tr \EE Q(0,0)(z)$ with more and more
sophisticated statistical model for $H(0,0)$. Among these, centered models
with the so-called
single-sided or double-sided correlations were studied in
\cite{chua-etal-it02, mes-fon-pag-jsac03, mou-sim-sen-it03, hac-etal-it08}.
A more general centered model was considered in \cite{tul-loz-ver-it05}.
Non-centered models were considered in \cite{mou-sim-jphys05,
tar-it08, dum-etal-it10} among others.
One contribution of our paper is that we consider frequency selective channels
and a model for the Doppler effect more realistic than the block-fading
model. \\

Ergodic operator theory (see \emph{e.g.}~\cite{pas-fig}) has aroused a
considerable interest in the fields of operator theory, quantum physics, and
statistical mechanics where these operators are frequently used to model the
Hamiltonians of randomly disordered systems.
In this framework, the contribution~\cite{kho-pas-cmp93} studied the
asymptotic behavior of the IDS of certain random Hermitian operators when
some design parameter is made converge to infinity. For the models considered
in~\cite{kho-pas-cmp93}, it is shown that the limiting IDS coincides with
that of a deformed large Wigner matrix, well known in random matrix theory.
The connections between ergodic operator theory and random matrix theory
are further explored in~\cite{pas-aa11}.
In this paper, we adopt the general approach of~\cite{kho-pas-cmp93,pas-aa11},
dealing with random operator models \`a la Marchenko-Pastur. In order to
perform
our asymptotic analysis, we rely on two tools that belong to the arsenal of
random matrix theory, as shown in~\cite{pas-livre} and the references therein:
an integration by parts formula for the expectation of functions of Gaussian
vectors (already used in~\cite{kho-pas-cmp93}), and the Poincar\'e-Nash
inequality to control our variances. \\

In the communication theory literature, the ergodic operator formalism for
studying the mutual information of multi-antenna channels was used in
\cite{kaf-isit02,kaf-ssp05}. These contributions brought to the fore the
relation between the mutual information and the IDS without elaborating on
the properties of the latter. \\

In the next section of this paper, we formulate precisely our assumptions and
we state our results. These results are illustrated in 
Section~\ref{sec:Numerics} by some computer simulations. An overview of 
the ergodic operator theory is provided in Section~\ref{ergo-op}. 
The proofs of the main results are provided in Sections~\ref{prf-det-eq} 
and~\ref{prf-approx}.

\section{Assumptions and the results}
\label{results}

We start by introducing the exact statistical model of our channel.
Since the number of antennas will tend to infinity, the channel will
be described by a sequence of random operators on $l^2(\ZZ)$ indexed by the
parameter $T$.

Given a probability space $(\Omega, \cal F, \PP)$ and two sequences
of positive integers $( N(T) )_{T\in\NN}$ and $(L(T))_{T\in\NN}$, consider
the sequence of random operators
\[
H_T = A_T + X_T
\]
with domain $\dom(H_T) = \cal K$, the manifold of vectors of $l^2(\ZZ)$
with finite support.
The operator $A_T$ is a deterministic operator whose matrix represented
in the canonical basis of $l^2(\ZZ)$ by the matrix
$A_T = [ A_T(k-\ell) ]_{k,\ell\in\ZZ}$ where the block $A_T(k)$ is a
$N(T) \times T$ complex matrix, and where the sequence
$(\ldots, A_T(-1), A_T(0), A_T(1),\ldots$ is supported by
$\{ -L(T), \ldots, L(T) \}$. Since the blocks of $A_T$ are identical along
the diagonals, $A_T$ is a convolution operator~\cite{bot-sil-toep}. \\
The operator $X_T$ is a 
$\cal K$-defined random operator on
$l^2(\ZZ)$ (in the sense of~\cite[Sec.~I.1.B]{pas-fig}, see
Section~\ref{ergo-op} below) whose matrix representation in the canonical
basis of $l^2(\ZZ)$ is the band matrix
\begin{align*}
X_T &=
\left[\begin{array}{lllllll}
\ddots  &         & \ddots             & \ddots & \ddots
& \!\!\!\!\!\!\!\!\!\!\!\!\!\ddots & 0 \\
        & \ddots  & X_T(0,-1) & X_T(0,0) & X_T(0,1) &
\ddots & \\
  0     &                    & \ddots & \!\!\!\!\!\!\ddots
       & \!\!\!\!\!\!\ddots & \!\!\!\!\!\!\ddots &  \ddots
\end{array}\right] \\
&= \begin{bmatrix} X_T(k,\ell),
k,\ell \in \ZZ, |k-\ell| \leq L(T) \end{bmatrix} .
\end{align*}
The $N(T) \times T$ random complex matrix $X_T(k,\ell)$ is written as
\[
X_T(k,\ell) = [ X_{T,n,t}(k,\ell), n=0:(N(T)-1), t =0:(T-1) ]
\]
and satisfies
\[
X_T(k,\ell) = \frac{1}{\sqrt{T}} \phi_T(k-\ell) W_T(k,\ell) .
\]
Here $\phi_T : \ZZ \to [0,\infty)$ is a real valued function supported by the
set $\{ -L(T),\ldots, L(T) \}$, and
$( W_T(k,\ell) =
[W_{T,n,t}(k,\ell), n=0:(N(T)-1), t =0:(T-1)])_{k,\ell \in\ZZ}$ is a complex
Gaussian circularly symmetric centered random field such that
\[
\EE [ W_{T,n_1, t_1}(k_1, \ell_1)
 \bar W_{T,n_2, t_2}(k_2, \ell_2) ] =
\delta_{n_1, n_2} \delta_{t_1 , t_2} \delta_{k_1 - \ell_1 , k_2 - \ell_2}
\gamma_T(k_1 - k_2)
\]
where the covariance function $\gamma_T(k)$ satisfies $\gamma_T(0) = 1$ without
generality loss.

The operator $A_T$ models the deterministic part of the channel. The
function $\phi_T$ models the multipath amplitude profile, while the
covariance function $\gamma_T$ is related to the Doppler effect. As is well
known, the function $\gamma_T(k)$ converges quickly to zero when the mobile
terminal velocity is high.

The following set of assumptions will be needed:
\begin{enumerate}
\item\label{as-reg}
$\displaystyle{ 0 < \liminf_{T\to\infty} N(T)/T \leq
\sup_{T} N(T)/T < \infty}$. We write $\bs c = \sup_T N(T)/T$.

\item\label{power}
$\displaystyle{\sigma_T^2 = \sum_{\ell} \phi_T(\ell)^2}$ satisfies
$0 < \liminf_T \sigma^2_T \leq \sup_T \sigma^2_T < \infty$.
We let $\bs\sigma^2 = \sup_T \sigma^2_T$.

\item\label{Tcoh}
$\displaystyle{\bs g = \sup_T \sum_{\ell} | \gamma_T(\ell) |}$
is finite. 

\item\label{bound-A}
$\displaystyle{\bs a = \sup_T \sum_{\ell=-L(T)}^{L(T)} \| A_T(\ell) \|}$
is finite, where $\|\cdot\|$ is the spectral norm.

\end{enumerate}

The asymptotic regime described by the above assumption will be concisely
denoted ``$T\to\infty$''.
Notice that there is no bound on the number of channel coefficients
$2L(T)+1$.

Let us comment these assumptions.  The parameter $\sigma_T^2$ in
Assumption~\ref{power} is the part of the received power due to the centered
part of the channel.  The practical interpretation of Assumption~\ref{Tcoh} is
that the so-called coherence time of the
channel~\cite{Biglieri1998_FadingChannels} does not grow with $T$. At this
point, we stress the importance of the fact that our results on the mutual
information assume the channel to be perfectly known at the receiver. In
practice, this channel needs to be estimated, and this task will be getting
harder as $T$ grows if the coherence time is kept fixed (see however the
discussion around Fig.~\ref{fig:MutualInfo_rho_num} below).  We believe that
relaxing Assumption~\ref{Tcoh} requires other mathematical tools than those
used in this paper.  Finally, Assumption~\ref{bound-A} could be certainly
weakened and replaced with \emph{e.g.}~a bound on the Euclidean norms of the
columns and the rows of the matrices $A(\ell)$ at the expense of a more
involved proof for Theorem~\ref{approx} below. We also point out that the
convergence rates obtained in the proof of this theorem can be improved.  We
chose to keep Assumption~\ref{bound-A} and the present proof for simplicity.
Note also that we assume the Gaussian distribution of various random variables
processes above largely for the simplicity of the corresponding proofs. By
using a version of the so-called interpolation trick (see \cite[Sections
18.3--18.4]{pas-livre}), our results can be extended to the case where the
corresponding ergodic random processes are not necessarily Gaussian just having
the same covariance and with a certain number finite moments and sufficiently
fast decaying correlations, although the respective proofs of the results
become rather involved.  \\

As mentioned in the introduction, the Stieltjes transform of a probability
measure $\nu$ on $\RR$ plays a fundamental role in this paper. This is the
function
\[
\bs s(z) = \int \frac{1}{\lambda-z} \nu(d\lambda)
\]
defined on $\CC_+$. This function is \emph{i)}~holomorphic on
$\CC_+$, \emph{ii)}~it satisfies $\bs s(z) \in \CC_+$
for any $z \in \CC_+$, and
\emph{iii)}~$\lim_{y\to\infty} (-\imath y) \bs s(\imath y) = 1$.
In addition, if $\nu$ is supported by $\RR_+ = [0,\infty)$, then
\emph{iv)}~$\Im(z \bs s(z)) \geq 0$ for any $z \in \CC_+$.
Conversely, any function $\bs s(z)$ satisfying
\emph{i)}--\emph{iv)} is the Stieltjes Transform of a probability measure
supported by $\RR_+$ \cite{Krein77}.
Observe that the Stieltjes Transform of $\nu$ can be trivially extended
from $\CC_+$ to $\CC - \support(\nu)$ where $\support(\nu)$ is the support of
$\nu$.
Finally, ``probability measure'' can be replaced with ``positive measure
$\nu$ such that $0 < \nu(\RR) < \infty$'' in the preceding statements if
we replace \emph{iii)} with
$\lim_{y\to\infty} (-\imath y) \bs s(\imath y) = \nu(\RR)$.

We shall now state our results.
Due to the fact that $\gamma_T$ is summable, the Gaussian stationary
random process
\[
\{{\bs H}_T(k) = [ H_T(k,k-L(T)), \ldots, H_T(k, k+L(T)) ]\}_{k\in\ZZ}
\]
is ergodic. This makes the operator $H_T$ ergodic in a sense made precise
below. The characterization of the mutual information in the framework of
ergodic operators is given by the two following results.

\begin{proposition}
\label{ids}
On a probability one set, the operator $H_T$ is closable. Still denoting
by $H_T$ the closure of this operator and by $H_T^*$ its adjoint, the positive
self-adjoint operator $H_T H_T^*$ is ergodic and has an IDS defining a
probability measure $\mu_T$. Denoting by $Q_T(z) = (H_T H_T^* - z I)^{-1} =
[ Q_T(k,\ell)(z) ]_{k,\ell\in\ZZ}$ the resolvent of $H_T H_T^*$ for
$z\in\CC_+$, where the blocks $Q_T(k,\ell)(z)$ are $N\times N$ matrices,
the Stieltjes transform of $\mu_T$ is $\bs m_T(z) = N^{-1} \EE\tr Q_T(0,0)(z)$.
\end{proposition}

\begin{theorem}
\label{I-ids}
The sequence $\{I( S^n; (Y^n, H^n))/(2n+1)\}$ converges as $n\to\infty$,
and its limit is the mutual information
$\displaystyle{I_T(S; (Y,H)) =
N \int \log(1+\lambda) \, \mu_T(d\lambda) < \infty}$.
\end{theorem}
Let
\[
\bs\gamma_T(f) = \sum_k \exp(2\imath\pi k f) \gamma_T(k)
\]
be the Fourier transform of the sequence $\{\gamma_T(k)\}$, and
\[
{\bs A}_T(f) = \sum_{k} \exp(2\imath\pi k f) A_T(k)
\]
be the $N\times T$ Fourier transform of the sequence $\{A_T(k)\}$. Write
$({\bs A}_T{\bs A}_T^*)(f) = {\bs A}_T(f){\bs A}_T^*(f)$ and
$({\bs A}_T^*{\bs A}_T)(f) = {\bs A}_T^*(f){\bs A}_T(f)$ for compactness.

\begin{theorem}
\label{det-eq}
Let $\SS_T(f,z)$ and $\tSS_T(f,z)$ be respectively the $N\times N$ and
$T\times T$ matrices
\begin{align}
\SS_T(f,z) &= \left[
-z \bigl( 1 + \sigma_T^2 {\bs\gamma}_T(f) \star \tbd_T(f,z)\bigr) I_N
+ \bigl( 1 + \sigma_T^2{\bs\gamma}_T(-f) \star \bd_T(f,z) \bigr)^{-1}
({\bs A}_T{\bs A}^*_T)(f) \right]^{-1}, \label{eq-S} \\
\tSS_T(f,z) &= \left[
-z \bigl( 1 + \sigma_T^2 {\bs\gamma}_T(-f) \star \bd_T(f,z)\bigr) I_T
+ \bigl( 1 + \sigma_T^2{\bs\gamma}_T(f) \star \tbd_T(f,z) \bigr)^{-1}
({\bs A}_T^*{\bs A}_T)(f) \right]^{-1}, \label{eq-tS}
\end{align}
and
\begin{align*}
{\bs\gamma}_T(f) \star \tbd_T(f,z) &=
   \int_0^1 {\bs\gamma}_T(f-u) \tbd_T(u,z) \, du, \ \text{and} \\
{\bs\gamma}_T(-f) \star \bd_T(f,z) &=
   \int_0^1 {\bs\gamma}_T(u-f) \bd_T(u,z) \, du .
\end{align*}
Then for any $z \in \CC_+$, the system of equations
\[
\bd_T(f,z) = \frac{\tr \SS_T(f,z)}{T}
\quad \text{and} \quad
\tbd_T(f,z) = \frac{\tr \tSS_T(f,z)}{T}
\]
admits a unique solution $( \bd_T(\cdot,z), \tbd_T(\cdot,z) )$ such that
$ \bd_T(\cdot,z), \tbd_T(\cdot,z) : [0,1] \to \CC$ are both measurable and
Lebesgue-integrable on $[0,1]$ and such that
$\Im \bd(f,z)$, $\Im \tbd(f,z)$, $\Im (z\bd(f,z))$ and $\Im (z\tbd(f,z))$ are
nonnegative for any $f\in[0,1]$. \\
The solutions $\bd_T(\cdot,z)$ and $\tbd_T(\cdot,z)$ are continuous on
$[0,1]$, and $\Im \bd(f,z)$, $\Im \tbd(f,z)$, $\Im (z\bd(f,z))$ and
$\Im(z\tbd(f,z))$ are positive for any $f\in[0,1]$. Furthermore,
for any $f\in[0,1]$, the functions $(T/N)\bd(f,z)$ and $\tbd(f,z)$ are
defined on $\CC-\RR_+$ and are the Stieltjes transforms of probability
measures supported by $[0,\infty)$.
\end{theorem}

This system of equations turns out to be formally close to that described
in \cite[Th.~2.4]{hachem-loubaton-najim07}. The uniqueness established in
the present paper is a point-wise uniqueness (\emph{i.e.}, for any $z\in\CC_+$)
which is stronger than the type of uniqueness shown
in~\cite{hachem-loubaton-najim07}. The proof of Theorem~\ref{det-eq} is
given in Section~\ref{prf-det-eq}. \\

The discussion preceding Proposition~\ref{ids} shows that the function
\[
\bs p_T(z) = \frac 1N \int_0^1 \tr \SS_T(f,z) \, df
\]
is the Stieltjes transform of a probability measure carried by $[0,\infty)$. We
denote it $\bs\pi_T$.

\begin{theorem}
\label{approx}
For any $z \in \CC_+$,
\begin{equation}
\label{cvg-st}
\bs m_T(z) - \bs p_T(z)
\xrightarrow[T\to\infty]{} 0 .
\end{equation}
Moreover, the sequences $\{\mu_T\}$ and $\{\bs\pi_T\}$ are tight, and
\[
\int g(\lambda) \, \mu_T(d\lambda) \ - \
\int g(\lambda) \, \bs\pi_T(d\lambda)
\xrightarrow[T\to\infty]{} 0
\]
for any continuous and bounded real function $g$.
\end{theorem}
The approximation of the mutual information for large $T$ is finally provided
by the following theorem.
\begin{theorem}
\label{capa}
We have
\[
N^{-1} I_T(S; (Y,H)) - {\cal I}_T \xrightarrow[T\to\infty]{} 0 ,
\]
where
\[
{\cal I}_T = \int \log(1+\lambda) \, \bs\pi_T(d\lambda) .
\]
and the integral is given by
\begin{align*}
{\cal I}_T &=
\frac 1N \int_0^1 \log\det\Bigl(
(1 + \sigma_T^2 {\bs\gamma}_T(f) \star \tbd_T(f,-1) ) I_N
+ \frac{({\bs A}_T{\bs A}^*_T)(f)}
{1 + \sigma_T^2{\bs\gamma}_T(-f) \star \bd_T(f,-1)} \Bigr) \, df  \\
&\phantom{=}
+ \frac TN \int_0^1 \log( 1 + \sigma_T^2{\bs\gamma}_T(-f) \star \bd_T(f,-1) )
\, df \\
&\phantom{=}
- \frac TN \int_0^1\int_0^1 \sigma_T^2 {\bs\gamma}_T(f-v)
   \tbd_T(v,-1) \bs\varphi_T(f,-1) \, dv \, df .
\end{align*}
\end{theorem}

Before turning to the numerical illustrations of these results, two remarks
are in order. We first observe that the variance profile represented by the
function $\phi_T^2 : \ZZ \to [0,\infty)$ has no influence on ${\mathcal I}_T$ 
except through the total received power $\sigma_T^2$ due to the random part
of the channel. 

We also observe that if the channel is centered, \emph{i.e.}, 
if $A_T = 0$, then $\bs \pi_T$ is the Marchenko-Pastur distribution 
\[
\bs\pi_T(d\lambda) = \frac{1}{2\pi c_T \sigma_T^2 \lambda} 
\sqrt{ (\lambda_+ - \lambda) (\lambda - \lambda_- ) }
\1_{[\lambda_-, \lambda_+]}(\lambda) \, d\lambda  
          + ( (1-c_N^{-1}) \vee 0) \delta_0 
\]
where $c_T = N/T$, $\lambda_+ = \sigma_T^2(1+\sqrt{c_T})^2$, and  
$\lambda_- = \sigma_T^2(1-\sqrt{c_T})^2$. It is indeed well known that 
the Stieltjes Transform of $\bs\pi_T$ is $c_N^{-1} \alpha(z)$, where 
$\alpha(z)$ is the unique solution in $\CC_+$ of the equations
\[
\alpha(z) = c_N ( -z -z \sigma_T^2 \tilde\alpha(z) )^{-1}, \quad
\tilde\alpha(z) = ( -z -z \sigma_T^2 \alpha(z) )^{-1} .
\]
Recalling that $\int {\bs\gamma}_T(f)\, df = 1$, one can immediately check that 
when $A_T = 0$, the couple $(\bs\varphi_T(f,z), \tilde{\bs\varphi}_T(f,z)) = 
(\alpha(z), \tilde\alpha(z))$ is a solution of the system described in the 
statement of Theorem~\ref{det-eq}. The result follows by the uniqueness of
this solution.


\section{Numerics}
\label{sec:Numerics}

In this section we will present the behavior of the mutual information for some
representative cases and will compare with numerically generated
instantiations. To begin with we present the usually accepted model for the
temporal correlation of fast fading, \emph{i.e.}, the so-called Jakes model
\cite{Jakes_book}, which, in the time domain has the following correlation
form:
\[
  \gamma_T(t)=J_0\left(\frac{2\pi v t}{\lambda}\right)
\] 
which in the frequency domain becomes
\[
  \gamma_T(f)=\frac{1}{\pi}\frac{1}{\sqrt{f_d^2-f^2}}
\] 
within the region $|f|<f_d$ and zero elsewhere, where $f_d=v/\lambda\tau$, the
ratio of the velocity of the mobile to the wavelength times the time duration
of the channel usage.  It should be noted that the discontinuity at $f_d$
results to $\gamma_T(t)$ not being absolutely summable, and hence strictly
speaking it cannot be used in this paper. However, demanding the frequency
response to be continuous at $f=f_d$, (by adding a small rounding factor in the
frequency domain), we can make this to be an acceptable model for the system.
For simplicity we will not include this in the simulations.

To move on, we need to also present a model for the deterministic matrix
function ${\bs A}_T(f)$. The typical situation for wireless communications is
that the constant matrices are due to line-of sight rank-one components. Along
these lines we assume that each of the time resolvable paths have the following
matrix elements
\[
A_T(k)_{m,n}= \exp\left[-|k| \xi/L_{tot}\right] \exp\left[2\pi j (m-n) \sin\theta_k\right]/\sqrt{N} 
\]
where $\theta_k = k \pi/L_{tot}$, for $k=-L,\ldots,L$ and $L_{tot}=2L+1$. The
exponential  dependence on the delay spread $\xi$ has been seen experimentally
\cite{Pedersen2000_SpaceTimeChannelModel, Calcev2004_3GPP_SCM}. Then 
${\bs A}_T(f)$ follows from
\[
{\bs A}_T(f)= \sum_{\ell=-L}^L A_T(k) \exp\left[-2\pi j \ell f\right]
\]
Note that the rank of the above matrix is $L_{max}$.

In the next figures, we observe the behavior of the mutual information as a
function of the Signal to Noise Ratio (SNR) parameter 
\[
\rho = \sigma_T^2 + \frac 1N \int_0^1 \tr ({\bs A}_T{\bs A}_T^*)(f) \, df 
\]
in various cases. In Fig.~\ref{fig:MutualInfo_rho_variousK} we plot the mutual
information for different values of the Ricean coefficient $K$ defined as the
ratio between the power of the deterministic part channel  to that of the of
the random part of the channel for the case. We see that for increasing values
of $K$ the mutual information initially increases because the deterministic
channel provides higher throughput, but then decreases because for higher $K$
the rank-deficiency of the deterministic channel becomes apparent.

In Fig. \ref{fig:MutualInfo_rho_variousXi} we see that the mutual information decreases with increasing inverse delay spread $\xi$, which is essentially the frequency correlation bandwidth. This behavior is expected, because the smaller $\xi$ is the smaller the attenuation and thus the higher the strength of the channel at the receiver.

We have also found that the dependence of the mutual information on the value
of $f_d$ is rather benign and not discernable for the above type of
deterministic channel. This can be attributed to the fact that the dependence
of the eigenvalue distribution of $({\bs A}_T {\bs A}_T^*)(f)$ on the frequency
value $f$ is quite small. For other less realistic examples of deterministic
channels, \emph{e.g.}, the case where $A(k)=\exp(-|k|\xi) I_N$, where $I_N$ is 
the identity matrix, does depend strongly on $f_d$.

Finally, in Fig. \ref{fig:MutualInfo_rho_num} we compare the results obtained
from the methodology presented in this paper with numerically generated values.
In this figure we plot the average mutual information generated numerically for
various block sizes (here depicted by $M$). In this case for simplicity, we
have used the exponentially decaying temporal correlation model,\emph{i.e.},
with $\gamma_T(k)=\exp(-|k|f_d)$. This model is easier to implement numerically
and absolutely summable, but is not very realistic because it is very wide-band.
The first important observation here is that this approach gives near exact
results for antenna arrays as small as $N=T=2$ presented in this figure.
Second, we see that depending on the value of $f_d$ the convergence to the
asymptotic result depends on the actual block size. For small $f_d$ (temporally
correlated channels) the $M=5$ set of points (diamonds) is significantly
deviating from the analytic curve, while for large $f_d$  the $M=5$ simulations
are much closer to the asymptotic result. Increasing the block size to $M=40$
makes the numerical values right on the analytical ones.

\begin{figure}[ht]
{\includegraphics*[width=\columnwidth]{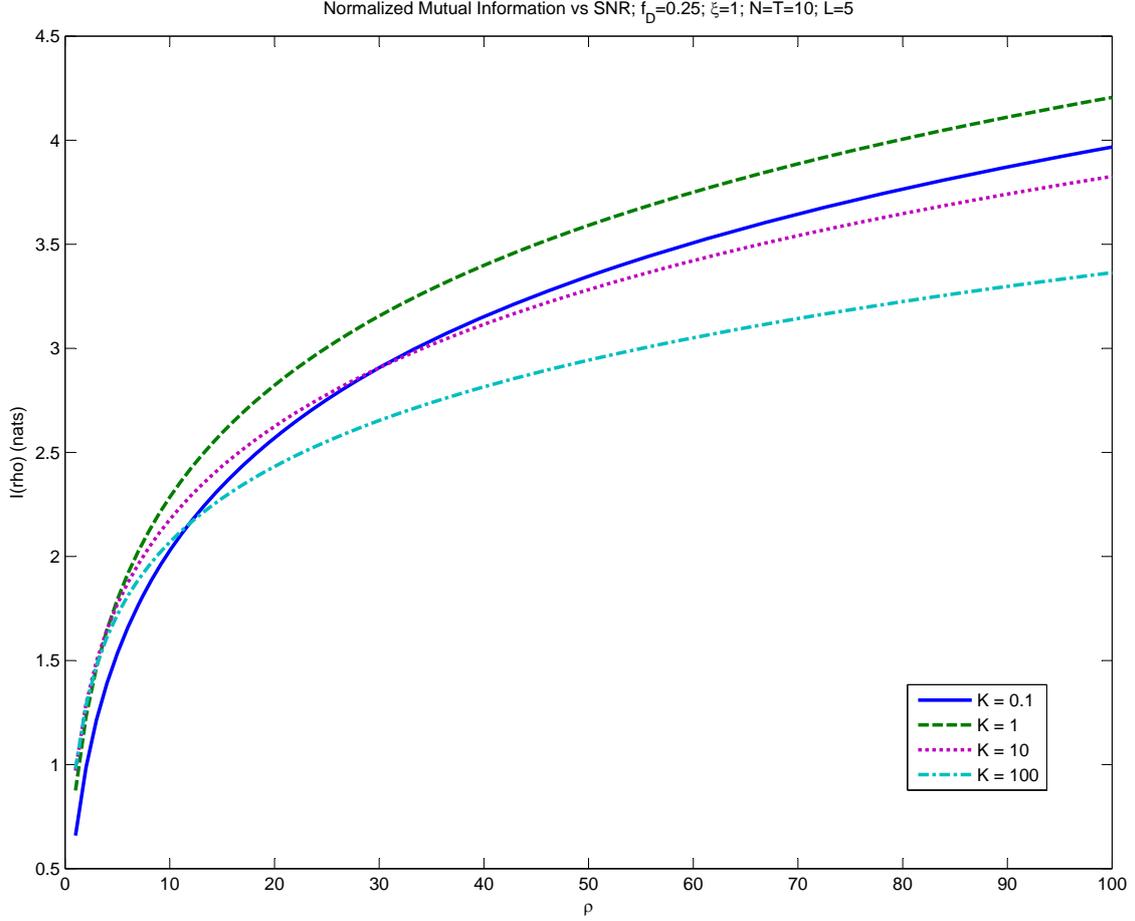}}
\caption{Mutual Information versus $\rho$ for various $K$}
\label{fig:MutualInfo_rho_variousK}
\end{figure}


\begin{figure}[ht]
{\includegraphics*[width=\columnwidth]{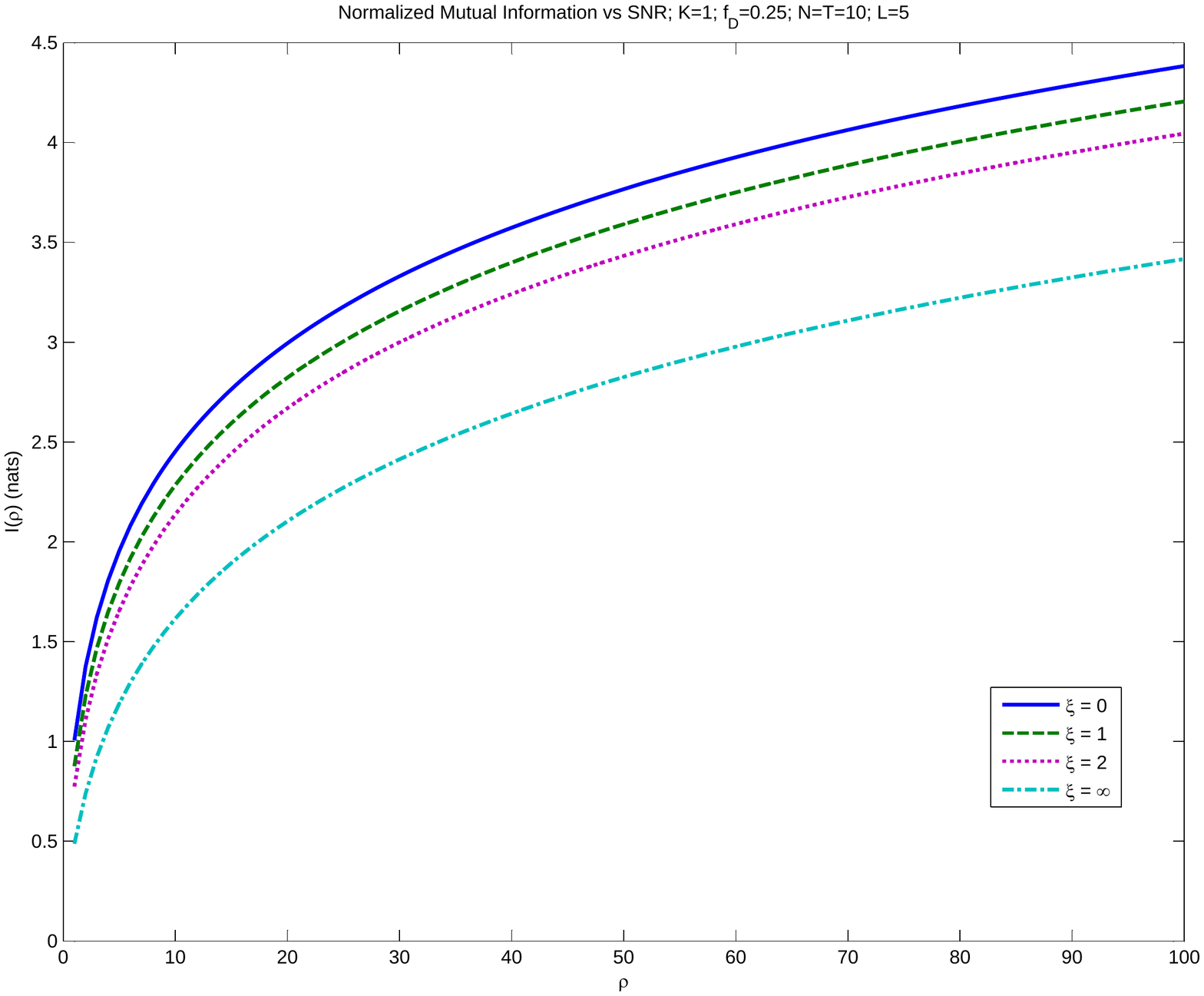}}
\caption{Mutual Information versus $\rho$ for various $\xi$}
\label{fig:MutualInfo_rho_variousXi}
\end{figure}

\begin{figure}[ht]
{\includegraphics*[width=\columnwidth]{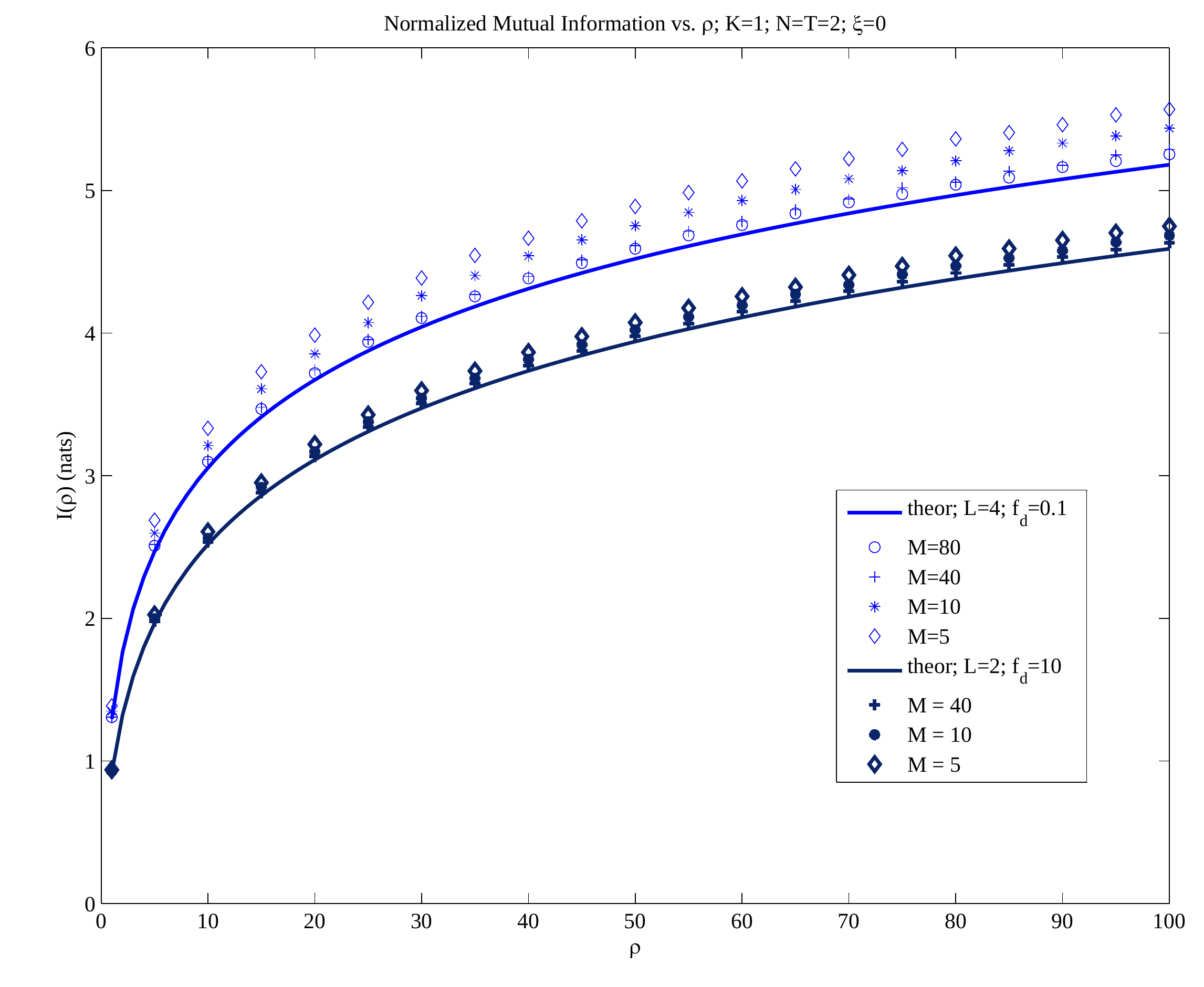}}
\caption{Mutual Information versus $\rho$ for various $L$ and block sizes}
\label{fig:MutualInfo_rho_num}
\end{figure}


\section{An overview of ergodic operators. Proofs of
Proposition~\ref{ids} and Theorem~\ref{I-ids}}
\label{ergo-op}

To make this paper reasonably self-contained, we provide a brief review
of the ergodic operator theory that will be useful to our purpose.
In passing, we shall prove Proposition~\ref{ids}. For a comprehensive
exposition of this theory, the reader may consult the book~\cite{pas-fig}.

Since the operator $H_T$ is defined for every $\omega\in\Omega$
on the dense linear subspace $\cal K$ of $\l^2(\ZZ)$, and since the vector
$H_T a$ is a random vector in $l^2(\ZZ)$ for any $a \in \cal K$, the operator
$H_T$ is random, following the definition of~\cite[Sec.~I.1.B]{pas-fig}.

For the presentation simplicity, let us temporarily assume that $N = T = 1$
and suppress the index $T$. Since $\cal K$ is dense, $H$ has an adjoint
$H^*$ with domain
\begin{align*}
\dom(H^*) &=
\Bigl\{ a \in l^2(\ZZ) \ : \
\sum_k | \langle H e_k, a \rangle |^2 < \infty \Bigr\}  \\
&=
\Bigl\{ a = \sum_{k\in\ZZ} \alpha_k e_k \ : \
\sum_{k\in\ZZ} \Bigl| \sum_{\ell=-L}^L H(k+\ell,k)
\bar\alpha_{k+\ell} \Bigr|^2 < \infty \Bigr\}
\end{align*}
where $e_k$ is the $k^{\text{th}}$ canonical basis vector of $l^2(\ZZ)$.
Clearly ${\cal K} \subset \dom(H^*)$, so this domain is dense in $l^2(\ZZ)$.
Therefore, $H$ is closable, and its closure coincides with $H^{**}$. In
the remainder, we reuse the notation $H$ to designate this closure.
Note that with this extension, the set $\cal K$ becomes a core for $H$
\cite[Chap. VIII]{ree-sim-1}. Since $H$ is closed and densely defined,
the operator $H H^*$ is a positive self-adjoint operator
with $\cal K$ as a core~\cite[Th.~46.2]{akh-glaz-93}. \\

We shall now consider the ergodic properties of $HH^*$. Making explicit the
dependence on the elementary event $\omega$ in the definition of the process
\[
{\bs H}(\omega) = \{{\bs H}(\omega,k) \}_{k\in\ZZ},
\quad
{\bs H}(\omega,k) = [ H(k,k-L), \ldots, H(k, k+L) ],
\]
it is easy to see that the summability of $\gamma$ implies that the shift
$B\,:\,\Omega\to\Omega$ defined by the equation
${\bs H}(B\omega, k) = {\bs H}(\omega, k+1)$ is ergodic. Moreover,
the random operator $H(\omega)$ clearly satisfies the equation
\[
H(B\omega) = U H(\omega) U^{-1}
\]
where $U$ is the (unitary) shift operator $Ua = \sum_k \alpha_{k+1} e_{k}$ for
$a = \sum_k \alpha_k e_k \in l^2(\ZZ)$. An operator on $l^2(\ZZ)$ that
satisfies such an equation where $B$ is ergodic is said ergodic in the sense
of~\cite[page 33]{pas-fig}.  Now, writing $(HH^*)(\omega) = H(\omega)
H^*(\omega)$,  we also have  $(HH^*)(B\omega) = U (HH^*)(\omega) U^{-1}$.  In
addition, it is clear that $U{\cal K} = {\cal K}$. Therefore, the operator
$HH^*$ is also ergodic in the sense of \cite[page 33]{pas-fig}.

Consider now the sequence of finite dimensional matrices
$( H^n H^{n*} )_{n\in\NN}$ obtained by truncating the matrix $HH^*$ and
keeping the elements $(i,j)$ such that $|i|,|j|\leq n$, see Eq.~\eqref{HnHn}.
Let $\lambda_{-n,n},\ldots,\lambda_{n,n}$ be the eigenvalues of $H^nH^{n*}$
and let
\[
\nu^n(\lambda) = \frac{1}{2n+1} \sum_{i=-n}^n \delta_{\lambda_{i,n}}
\]
be the Normalized Counting Measure (NCM) of this matrix. Since the operator
$HH^*$ is represented by a band matrix, we have the following result, stemming
directly from \cite[Th.~II-4.8]{pas-fig}:
\begin{proposition}
\label{pas-fig}
The operator $HH^*$ has an IDS. In other words, there exists a deterministic
probability measure $\mu$ on $[0,\infty)$ such that
\begin{equation}
\label{C-ids}
\int g(\lambda) \nu^n(d\lambda) \toaslong \int g(\lambda) \mu(d\lambda)
\end{equation}
for all continuous and bounded functions $g$. Moreover, the Stieltjes
transform of $\mu$ coincides with $\EE Q(0,0)(z)$.
\end{proposition}
Proposition~\ref{ids} for $N = T = 1$ follows. It will be useful to outline
the proof of Proposition~\ref{pas-fig}. The idea is to
establish the convergence in the statement of this proposition for a large
enough family of continuous functions. Given a continuous function $g$, start
writing the left hand side of Eq.~\eqref{C-ids} as
\[
\int g(\lambda) \nu^n(d\lambda) =
\frac{1}{2n+1} \sum_{i=-n}^n g(H^nH^{n*})(i,i)
\]
where $g(H^nH^{n*})(i,i)$ is the element $(i,i)$ of the matrix
$g(H^nH^{n*})$. Suppose that for large $n$, the operator at the right hand 
side of this equation can be replaced with $HH^*$. 
Denote by $g(HH^*)(i,i)$ the $(i,i)$ element of the matrix representation of 
$g(HH^*)$.
Observe now that since $HH^*$ is ergodic, then by \cite[Theorem 2.7]{pas-fig},
its \emph{resolution of the identity}  $E_\lambda$ is an ergodic projection
for every $\lambda \in\RR$. As a result, the operator
$g(HH^*) = \int g(\lambda) dE_\lambda$ is also ergodic. Thus,
\begin{align*}
\frac{1}{2n+1} \sum_{i=-n}^n g(HH^*)(\omega)(i,i) &=
\frac{1}{2n+1} \sum_{i=-n}^n \langle e_i, g(HH^*)(\omega) e_i \rangle \\
&=
\frac{1}{2n+1} \sum_{i=-n}^n
\langle e_0, U^i g(HH^*)(\omega) U^{-i} e_0 \rangle \\
&=
\frac{1}{2n+1} \sum_{i=-n}^n
\langle e_0, g(HH^*)(B^i \omega) e_0 \rangle \\
&\toaslong \EE g(HH^*)(0,0) =
\EE\int g(\lambda) \langle e_0, E(d\lambda) e_0 \rangle
\end{align*}
by the ergodic theorem, provided the $(0,0)$ element $g(HH^*)(0,0)$
of $g(HH^*)$ is integrable. It follows that $\mu(d\lambda) =
\EE \langle e_0, E(d\lambda) e_0 \rangle$. \\
By taking $g(\lambda) = (\lambda - z)^{-1}$ where $z \in\CC_+$, we get that
the Stieltjes transform of $\mu$ coincides with
$\EE g(HH^*)(0,0) = \EE Q_{HH^*}(0,0)(z)$.

It remains to generalize these results to the case where the $H_T(k,\ell)$
are $N \times T$ matrices. The closedness of $H_T$ and the existence of the
self-adjoint operator $H_TH_T^*$ are direct generalizations of the scalar case.
Writing $\bs H_T(\omega,k) = [ H_T(k,k+\ell)]_{\ell=-L(T)}^{L(T)}$, the
shift $B$ defined by the equation $\bs H_T(B\omega,k) =
\bs  H_T(\omega, k+1)$ is ergodic. Since the operator $H_T H_T^*$ satisfies
the identity
$(H_TH_T^*)(B\omega) = U^N (H_TH_T^*)(\omega) U^{-N}$, it is also ergodic.
The proof of Proposition~\ref{pas-fig} goes along nearly without modification.
Let us just check that the Stieltjes transform of $\mu_T$ is
$N^{-1} \tr \EE Q(0,0)(z)$.
Writing
$g(H_TH_T^*)_{p,q}(k,\ell) =
            \langle e_{kN+p}, g(H_TH_T^*) e_{\ell N+q} \rangle$,
we get by mimicking the derivation above that
\[
\frac{1}{(2n+1)N} \sum_{r=0}^{N-1}\sum_{i=-n}^n g(H_TH_T^*)_{r,r}(i,i)
\toaslong
\frac 1N \tr \EE g(H_TH_T^*)(0,0) .
\]
Taking $g(\lambda) = (\lambda-z)^{-1}$, we get the result.

\subsection*{Proof of Theorem~\ref{I-ids}}
As is well known (see \emph{e.g.} \cite{tel-95}),
\begin{align*}
I(S; (Y,H))
&= \limsup_n \frac{1}{2n+1} \Bigl( I(S^n; H^n) + I(S^n; Y^n \, | \, H^n)\Bigr)
\\
&= \limsup_n \frac{1}{2n+1} I(S^n; Y^n \, | \, H^n)  \\
&= \limsup_n \frac{1}{2n+1} \EE \log\det( H^n {H^n}^* + I_{(2n+1)N} )
\end{align*}
by the independence of $S^n$ and $H^n$.
Let $\nu^n_T$ be the NCM of the random matrix $H^n {H^n}^*$. Then
\[
\frac{1}{(2n+1)N} \log\det( H^n {H^n}^* + I ) = \int\log(1+\lambda) \,
      \nu^n_T(d\lambda) .
\]
The measure $\nu^n_T$ satisfies
\begin{align}
\int \lambda \, \nu^n_T(d\lambda) &= \frac{\tr (H^n {H^n}^*)}{(2n+1)N}
\nonumber \\
&=
\frac{1}{(2n+1)N} \sum_{m=-n}^n \sum_{\ell=-L}^L \tr H(m,m-\ell) H(m,m-\ell)^*
\nonumber \\
&\leq
\frac{2}{(2n+1)N} \sum_{m=-n}^n \sum_{\ell=-L}^L
(\tr A(\ell) A(\ell)^* + \tr X(m,m-\ell) X(m,m-\ell)^* ) \nonumber \\
&=
\frac 2N \sum_{\ell=-L}^L \tr A(\ell) A(\ell)^* \nonumber \\
&\phantom{=}
+  2 \sum_{\ell=-L}^L \frac{\phi_T(\ell)^2}{T}
\frac{1}{(2n+1)N}\sum_{m=-n}^n \tr W(m,m-\ell) W(m,m-\ell)^* .
\label{Enu}
\end{align}
Since the process $\{W(m,m-\ell) W(m,m-\ell)^*\}_{m\in\ZZ}$ is ergodic,
it holds that
\[
\frac{1}{2n+1}\sum_{m=-n}^n W(m,m-\ell) W(m,m-\ell)^*
\toaslong T \, I_N .
\]
We therefore get that
\[
\limsup_n
\int \lambda \, \nu^n_T(d\lambda) \leq
2( \sum_{\ell=-L}^L \| A_T(\ell) \|^2 + \sigma_T^2)
\quad \text{w.p. 1}.
\]
Let us consider an elementary event in the probability one set where
$\nu^n_T$  converges weakly to $\mu_T$ by Prop.~\ref{pas-fig} and where
$\sup_n \int \lambda d\nu^n_T < \infty$. 
By uniform integrability, we get that  
$\int \log(1+\lambda) \nu^n_T(d\lambda) \to
\int \log(1+\lambda) \mu_T(d\lambda) < \infty$ on the elementary event we
just considered. Consequently,
$\int \log(1+\lambda) \nu^n_T(d\lambda) \to
\int \log(1+\lambda) \mu_T(d\lambda) < \infty$ almost surely. \\
Getting back to the inequality~\eqref{Enu}, it is clear that
$\sup_n \EE \int \lambda \nu^n_T(d\lambda) \leq
2( \sum_\ell \| A_T(\ell) \|^2 + \sigma_T^2)$.
Therefore, $\sup_n \EE (\int \log(1+\lambda) \nu^n_T(d\lambda))^2 < \infty$
which shows that
\[
\EE \int \log(1+\lambda) \nu^n_T(d\lambda)
\xrightarrow[n\to\infty]{}
\int \log(1+\lambda) \mu_T(d\lambda)  < \infty
\]
by uniform integrability.

\section{Proof of Theorem~\ref{det-eq}}
\label{prf-det-eq}

We start by showing the uniqueness of the solution
$( \bd_T(\cdot,z), \tbd_T(\cdot,z) )$ satisfying the integrability and the
non-negativity conditions.
We then provide a constructive existence proof, and prove in passing that the
constructed solution satisfies the Stieltjes transform properties provided in
the statement.

\subsection{The uniqueness}
\label{unique}

The general idea of the uniqueness proof can be found in
\emph{e.g.}~\cite{DozSil07a} or~\cite{hac-ihp-bilin13}.
Applying this idea to the case of this article requires some specific work.
We fix $z \in\CC_+$ and we re-denote herein $\bd(f,z)$, $\SS(f,z)$, etc.
as $\bd(f)$, $\SS(f)$, etc.~for simplicity. \\
We start with some lemmas. Assume that $(\bd_T(f), \tbd_T(f) )$ is a solution.
Then
\begin{lemma}
\label{sys>0}
The solution $(\bd_T(f), \tbd_T(f) )$ satisfies the equation
\begin{align*}
\begin{bmatrix} \Im \bd(f) \\ \Im(z\tbd(f)) \end{bmatrix} &=
\int_0^1
\begin{bmatrix}
K_{11}^{\bd,\tbd}(f,u) & K_{12}^{\bd,\tbd}(f,u) \\
K_{21}^{\bd,\tbd}(f,u) & K_{22}^{\bd,\tbd}(f,u) \end{bmatrix}
\begin{bmatrix} \Im \bd(u) \\ \Im(z\tbd(u)) \end{bmatrix}
du \\
&\phantom{=}  + \Im z
\begin{bmatrix} F_{1}^{\bd,\tbd}(f) \\ F_{2}^{\bd,\tbd}(f) \end{bmatrix} ,
\quad f\in[0,1] ,
\end{align*}
where
\begin{align*}
K_{11}^{\bd,\tbd}(f,u) &=
    \frac{\tr \SS(f) ({\bs A}{\bs A}^*)(f) \SS(f)^*}
      {T |1+\bz(f)|^2} \sigma^2 \bs\gamma(u-f), \\
K_{12}^{\bd,\tbd}(f,u) &=
    \frac{\tr \SS(f) \SS(f)^*}{T} \sigma^2 \bs\gamma(f-u) , \\
K_{21}^{\bd,\tbd}(f,u) &=
    \frac{\tr |z|^2 \tSS(f) \tSS(f)^*}{T} \sigma^2 \bs\gamma(u-f), \\
K_{22}^{\bd,\tbd}(f,u) &=
    \frac{\tr \tSS(f) ({\bs A}^*{\bs A})(f) \tSS(f)^*}
      {T |1+\tbz(f)|^2} \sigma^2 \bs\gamma(f-u),
\end{align*}
\begin{align*}
F_{1}^{\bd,\tbd}(f) &= \frac{\tr \SS(f) \SS(f)^*}{T} , \\
F_{2}^{\bd,\tbd}(f) &=
    \frac{\tr \tSS(f) ({\bs A}^*{\bs A})(f) \tSS(f)^*}
      {T |1+\tbz(f)|^2} ,
\end{align*}
\[
\bz(f) = \sigma^2 \bs\gamma(-f)\star \bd(f), \quad\text{and}\quad
\tbz(f) = \sigma^2 \bs\gamma(f)\star \tbd(f) .
\]
\end{lemma}
\begin{proof}
We have
\begin{align*}
\Im \bd(f) &= \frac{\bd - \bar\bd}{2\imath} \\
&= \frac{1}{2\imath T} \tr \SS ( \SS^{-*} - \SS^{-1} ) \SS^* \\
&= \frac{1}{2\imath T} \tr \SS \Bigl( (z - \bar z) I  +
\sigma^2 \bs\gamma(f) \star ( z\tbd(f) - \bar z \bar\tbd(f) ) I \\
& \ \ \ \ \ \ \ \
+\frac{({\bs A}{\bs A}^*)(f)}{1+\bar\bz(f)}
-\frac{({\bs A}{\bs A}^*)(f)}{1+\bz(f)}
 \Bigr) \SS^* \\
&= \Im z \frac{\tr \SS(f) \SS^*(f)}{T}
+ (\Im(z\tbz(f)) ) \frac{\tr \SS(f)\SS^*(f)}{T} \\
&\phantom{=} + (\Im\bz(f))
\frac{\tr \SS(f) ({\bs A}{\bs A}^*)(f) \SS^*(f)}{T |1+\bz(f)|^2} .
\end{align*}
By a similar derivation, we also have
\begin{align*}
\Im(z\tbd(f)) &=
\frac{1}{2\imath T}
  \tr |z|^2 \tSS ( \bar z^{-1} \tSS^{-*} - z^{-1} \tSS^{-1} ) \tSS^* \\
&=
(\sigma^2 \bs\gamma(-f) \star \Im\bd(f) ) \frac{|z|^2\tr \tSS(f) \tSS^*(f)}{T}
+ \Im z \frac{\tr \tSS(f) ({\bs A}^*{\bs A})(f) \tSS^*(f)}{T |1+\tbz(f)|^2} \\
&\phantom{=}
+ (\sigma^2 \bs\gamma(f) \star \Im(z\tbd(f)) )
\frac{\tr \tSS(f) ({\bs A}^*{\bs A})(f) \tSS^*(f)}{T |1+\tbz(f)|^2}
\end{align*}
hence the result.
\end{proof}

\begin{lemma}
\label{bornes}
For any $f\in[0,1]$, $\Im\bd(f) > 0$, $F_{1}^{\bd,\tbd}(f) > 0$,
$0 < \Im(z\tbd(f)) < C$, and $F_{2}^{\bd,\tbd}(f) > C'
\tr ({\bs A}^*{\bs A})(f)$ where $C$ and $C'$ are positive numbers
that do not depend on $f$.
\end{lemma}
\begin{proof}
Let $\lambda_0,\ldots,\lambda_{N-1}$ be the eigenvalues of
$({\bs A}{\bs A}^*)(f)$. We start by establishing an
upper bound on $|-z(1+ \tbz) + \lambda_i/(1+\bz)|$.
Observing that $\max_f \bs\gamma(f) \leq \sum |\gamma_T(\ell)| \leq \bs g$
and writing $\tilde{\bs b}(z) = \int_0^1 | \tbd(f,z) | df < \infty$, we have
$| z(1+ \tbz) | \leq |z|( 1 + \bs\sigma^2 \bs g \tilde{\bs b}(z) )$. Moreover,
$|1+\zeta| = |z + z\zeta|/|z| \geq \Im z / |z|$ since $\Im(z\bd) \geq 0$ by
assumption, and $\bs\gamma(f) \geq 0$. Therefore,
$
| \lambda_i/(1+\bz) | \leq |z| \bs a^2 / \Im z
$
and $|-z(1+ \tbz) + \lambda_i/(1+\bz)| \leq
|z|( 1 + \bs\sigma^2 \bs g \tilde{\bs b}(z) + \bs a^2/\Im z )$.
The inequality $F_{1}^{\bd,\tbd}(f) > 0$ follows readily. Moreover,
\begin{equation}
\label{Imphi}
\Im\bd(f) =
\frac 1T \sum_{i=0}^{N-1}
\frac{\Im(-\bar z(1+ \bar\tbz) + \frac{\lambda_i}{1+\bar\bz})}
{|-z(1+ \tbz) + \frac{\lambda_i}{1+\bz}|^2}
\geq
\frac 1T \sum_{i=0}^{N-1}
\frac{\Im z}{|-z(1+ \tbz) + \frac{\lambda_i}{1+\bz}|^2} > 0 .
\end{equation}
The inequalities on $\Im(z\tbd(f))$ and $F_{2}^{\bd,\tbd}(f)$ are proven
by similar means.
\end{proof}
\begin{lemma}
\label{cont}
The functions $\bd$ and $\tbd$ are continuous on $[0,1]$.
\end{lemma}
\begin{proof}
Since $\bd$ is integrable and $\bs\gamma$ is continuous, $\bz$
is continuous by the dominated convergence theorem, and
similarly for $\tbz$. Moreover, $({\bs A}{\bs A}^*)$ is continuous as
a trigonometric polynomial. It follows that the matrix functions
$\SS$ and $\tSS$ are continuous, hence the result.
\end{proof}
To establish the uniqueness of the solution, we observe that by setting
\[
J_1^{\bd,\tbd}(f,u) =
\frac{K_{11}^{\bd,\tbd}(f,u)}{\Im \bd(f)} \Im\bd(u)
+ \frac{K_{12}^{\bd,\tbd}(f,u)}{\Im \bd(f)} \Im(z\tbd(u))
\]
and
\[
J_2^{\bd,\tbd}(f,u) =
\frac{K_{21}^{\bd,\tbd}(f,u)}{\Im(z\tbd(f))} \Im\bd(u)
+ \frac{K_{22}^{\bd,\tbd}(f,u)}{\Im(z\tbd(f))} \Im(z\tbd(u))
\]
on $[0,1]^2$, Lemmas~\ref{sys>0} and~\ref{bornes} show that
\begin{equation}
\label{J1}
0 < \int_0^1  J_1^{\bd,\tbd}(f,u,z) \, du
\leq 1 - \Im z \frac{F_{1}^{\bd,\tbd}(f,z)}{\Im \bd(f,z)} < 1
\end{equation}
and
\begin{equation}
\label{J2}
0 < \int_0^1  J_2^{\bd,\tbd}(f,u,z) \, du
\leq 1 - \Im z \frac{F_{2}^{\bd,\tbd}(f,z)}{\Im(z\tbd(f,z))}
\leq 1 - C \tr ({\bs A}^*{\bs A})(f)
\end{equation}
where $C > 0$.
Assume now that $(\bd, \tbd)$ and $(\bd', \tbd')$ are two solutions.
Denote respectively by $(\SS, \tSS)$ and $(\SS', \tSS')$ the matrix
functions associated to these solutions as in the statement of the proposition.
Write $\Delta\bd(f) = \bd(f) - \bd'(f)$ and
$z\Delta \tbd(f) = z\tbd(f) - z \tbd'(f)$.
Mimicking the proof of Lemma~\ref{sys>0}, we get
\begin{align*}
\Delta\bd(f)
&= \frac{1}{T} \tr \SS ( \SS'^{-1} - \SS^{-1} ) \SS' \\
&= \int_0^1 ( {\cal K}_{11}(f,u) \Delta\bd(u)
        + {\cal K}_{12}(f,u) z\Delta\tbd(u) ) \, du
\end{align*}
and
\begin{align*}
z\Delta\tbd(f)
&= \frac{z}{T} \tr \tSS ( \tSS'^{-1} - \tSS^{-1} ) \tSS' \\
&= \int_0^1 ( {\cal K}_{21}(f,u) \Delta\bd(u)
  + {\cal K}_{22}(f,u) z\Delta\tbd(u) ) \, du
\end{align*}
where
\begin{equation}
\label{Kij}
\begin{split}
\cal K_{11}(f,u) &=
    \frac{\tr \SS(f) ({\bs A}{\bs A}^*)(f) \SS'(f)}
      {T (1+\sigma^2\bs\gamma(-f)\star\bd(f))
          (1+\sigma^2\bs\gamma(-f)\star\bd'(f))} \sigma^2 \bs\gamma(u-f), \\
\cal K_{12}(f,u) &=
    \frac{\tr \SS(f) \SS'(f)}{T} \sigma^2 \bs\gamma(f-u) , \\
\cal K_{21}(f,u) &=
    \frac{\tr z^2 \tSS(f) \tSS'(f)}{T} \sigma^2 \bs\gamma(u-f), \\
\cal K_{22}(f,u) &=
    \frac{\tr \tSS(f) ({\bs A}^*{\bs A})(f) \tSS'(f)}
      {T (1+\sigma^2\bs\gamma(f)\star\tbd(f))
          (1+\sigma^2\bs\gamma(f)\star\tbd'(f))} \sigma^2 \bs\gamma(f-u) .
\end{split}
\end{equation}
Let
\[
\varepsilon(f) = \frac{\Delta\bd(f)}{\sqrt{ \Im\bd(f) \Im\bd'(f)}}
\quad \text{and} \quad
\tilde \varepsilon(f) = \frac{z\Delta\tbd(f)}
       {\sqrt{ \Im(z\tbd(f)) \Im(z\tbd'(f))}} .
\]
Using the inequality $|\tr(M_1M_2)| \leq \sqrt{\tr(M_1M_1^*)}
\sqrt{\tr(M_2M_2^*)}$ for any two matrices $M_1$ and $M_2$ with compatible
dimensions, we get
\begin{multline*}
|\varepsilon(f)| \leq
\int_0^1
\Bigl(\frac{\sqrt{K_{11}^{\bd,\tbd}(f,u) \Im\bd(u)}}{\sqrt{\Im\bd(f)}}
\frac{\sqrt{K_{11}^{\bd',\tbd'}(f,u) \Im\bd'(u)}}{\sqrt{\Im\bd'(f)}}
|\varepsilon(u)| \\
+
\frac{\sqrt{K_{12}^{\bd,\tbd}(f,u) \Im(z\tbd(u))}}{\sqrt{\Im\bd(f)}}
\frac{\sqrt{K_{12}^{\bd',\tbd'}(f,u) \Im(z\tbd'(u))}}{\sqrt{\Im\bd'(f)}}
|\tilde \varepsilon(u)| \Bigr) \, du
\end{multline*}
Let $\bs \varepsilon = \sup_{f\in[0,1]} (|\varepsilon(f)| \vee |\tilde \varepsilon(f)|)$. Using the
inequality $\sqrt{ac} + \sqrt{bd} \leq \sqrt{a+b} \sqrt{c+d}$ where
$a,b,c,d \geq 0$ along with the Cauchy-Schwarz inequality, we get
\[
|\varepsilon(f)| \leq \Bigl(\int_0^1 J_1^{\bd,\tbd}(f,u)\, du\Bigr)^{1/2}
\Bigl(\int_0^1 J_1^{\bd',\tbd'}(f,u)\, du\Bigr)^{1/2} \bs \varepsilon
\]
We shall assume that $\bs \varepsilon > 0$ and obtain a contradiction.
By the last inequality and~\eqref{J1}, we have $|\varepsilon(f)| < \bs \varepsilon$.
Lemmas~\ref{bornes} and \ref{cont} show that $\varepsilon(f)$ is continuous on $[0,1]$,
hence $\sup_{f\in[0,1]} |\varepsilon(f)| < \bs \varepsilon$.
Turning to $\tilde \varepsilon(f)$, we get by a similar argument that
\begin{multline*}
|\tilde \varepsilon(f)| \leq
\int_0^1
\Bigl(\frac{\sqrt{K_{21}^{\bd,\tbd}(f,u) \Im\bd(u)}}{\sqrt{\Im(z\tbd(f))}}
\frac{\sqrt{K_{21}^{\bd',\tbd'}(f,u) \Im\bd'(u)}}{\sqrt{\Im(z\tbd'(f))}}
|\varepsilon(u)| \\
+
\frac{\sqrt{K_{22}^{\bd,\tbd}(f,u) \Im(z\tbd(u))}}{\sqrt{\Im(z\tbd(f))}}
\frac{\sqrt{K_{22}^{\bd',\tbd'}(f,u) \Im(z\tbd'(u))}}{\sqrt{\Im(z\tbd'(f))}}
|\tilde \varepsilon(u)| \Bigr) \, du .
\end{multline*}
If  $({\bs A}^*{\bs A})(f) = 0$, then
$K_{22}^{\bd,\tbd}(f,u) = K_{22}^{\bd',\tbd'}(f,u) = 0$ for any $u\in[0,1]$.
But then, Inequalities~\eqref{J2} show that
$|\tilde \varepsilon(f)| \leq \sup_{u\in[0,1]} |\varepsilon(u)| < \bs \varepsilon$. On the other hand,
if $\tr ({\bs A}^*{\bs A})(f) >0$, then
\[
|\tilde \varepsilon(f)| \leq \Bigl(\int_0^1 J_2^{\bd,\tbd}(f,u)\, du\Bigr)^{1/2}
\Bigl(\int_0^1 J_2^{\bd',\tbd'}(f,u)\, du\Bigr)^{1/2} \bs \varepsilon
< \bs \varepsilon
\]
by~\ref{J2} again. Therefore, we also have
$\sup_{f\in[0,1]} |\tilde \varepsilon(f)| < \bs \varepsilon$ by the continuity
of $\tilde \varepsilon(f)$, which leads to a contradiction. Uniqueness is
established.

\subsection{The existence}
\label{exist}
Starting with the functions $\bd^{(0)}(f,z) = -(N/T) z^{-1}$ and
$\tbd^{(0)}(f,z) = -z^{-1}$ on $[0,1] \times \CC_+$,
define the recursion
\[
(\bd^{(k+1)}(f,z) , \tbd^{(k+1)}(f,z)) = h_z(\bd^{(k)}(f,z),\tbd^{(k)}(f,z))
\]
with
\begin{gather*}
 h_z(\bd^{(k)}(f,z), \tbd^{(k)}(f,z)) =
\Bigl(T^{-1} \tr \SS^{(k)}(f,z), T^{-1} \tr \tSS^{(k)}(f,z) \Bigr), \\
\SS^{(k)}(f,z) = \left[
-z ( 1 + \tbz^{(k)}(f,z)) I +
    \frac{({\bs A}{\bs A}^*)(f)}{1+\bz^{(k)}(f,z)}\right]^{-1} , \\
\tSS^{(k)}(f,z) = \left[
-z ( 1 + \bz^{(k)}(f,z)) I +
    \frac{({\bs A}^*{\bs A})(f)}{1+\tbz^{(k)}(f,z)}\right]^{-1}, \\
\bz^{(k)}(f,z) = \sigma^2 \bs\gamma(-f)\star \bd^{(k)}(f,z),
\quad \text{and} \quad
\tbz^{(k)}(f,z) = \sigma^2 \bs\gamma(f)\star \tbd^{(k)}(f,z).
\end{gather*}
We shall show that the sequence $(\bd^{(k)}(f,z) , \tbd^{(k)}(f,z))$
converges on $[0,1]\times\CC_+$, and that the limit $(\bd(f,z) , \tbd(f,z))$
is the solution of the system described in the statement of
Theorem~\ref{det-eq}.
\begin{lemma}
\label{recursion}
The following facts hold true:
\begin{itemize}
\item[i)] For any $k\in\NN$ and any $f\in[0,1]$, the functions
$\bd^{(k)}(f,\cdot)$ and $\tbd^{(k)}(f,\cdot)$ are holomorphic on $\CC_+$ and
satisfy $|\bd^{(k)}(f,z)| \leq \bs c/\Im z$ and
$|\tbd^{(k)}(f,z)| \leq 1/\Im z$. Furthermore,
$\Im \bd^{(k)}(f,z)$, $\Im (z\bd^{(k)}(f,z))$,
$\Im \tbd^{(k)}(f,z)$, and $\Im(z\tbd^{(k)}(f,z))$ are all positive for
$z\in\CC_+$.
\item[ii)] In the region
\[
{\cal R} = \Bigl\{ z \in\CC_+ \ : \
(\bs c\vee 1) \bs\sigma^2 \Bigl( \frac{\bs a^2 |z|^2}{(\Im z)^4}
+ \frac{|z|}{(\Im z)^2} \Bigr) < \frac 12
\Bigr\},
\]
it holds that
\begin{multline*}
| \bd^{(k+1)}(f,z) - \bd^{(k)}(f,z) | \vee
              | \tbd^{(k+1)}(f,z) - \tbd^{(k)}(f,z) |  \\
  < \frac 12
      ( | \bd^{(k)}(f,z) - \bd^{(k-1)}(f,z) | \vee
            | \tbd^{(k)}(f,z) - \tbd^{(k-1)}(f,z) | )  .
\end{multline*}

\item[iii)] For any $f, f' \in [0,1]$ and any $z\in\CC_+$,
\begin{multline*}
| \bd^{(k)}(f,z) - \bd^{(k)}(f',z) | \vee
              | \tbd^{(k)}(f,z) - \tbd^{(k)}(f',z) |  \\
\leq \frac{(\bs c\vee 1) |z|}{(\Im z)^3} \Bigl(
\Bigl( \bs\sigma^2 + \frac{|z| \bs\sigma^2\bs a^2}{(\Im z)^2} \Bigr)
\int_0^1 | \bs\gamma(f - f' + u) - \bs\gamma(u) | \, du
+ 4\pi L \bs a^2 | f - f' | \Bigr) .
\end{multline*}
\end{itemize}
\end{lemma}
\begin{proof}
The analyticity as well as the inequalities stated in \emph{i)} are trivially
true for $\bd^{(0)}$ and for $\tbd^{(0)}$. Assume they are for $\bd^{(k)}$
and $\tbd^{(k)}$. Denoting again by $\lambda_0,\ldots,\lambda_{N-1}$ the
eigenvalues of $({\bs A}{\bs A}^*)(f)$, we have
\begin{equation}
\label{den>Imz}
|-z(1+ \tbz^{(k)}) + \lambda_i/(1+\bz^{(k)})| \geq
\Im z + \Im (z\tbz^{(k)}) + \lambda_i \Im\bz^{(k)} /|1+\bz^{(k)}| \\
\geq \Im z .
\end{equation}
Consequently, $\bd^{(k+1)}$ and $\tbd^{(k+1)}$ are holomorphic on $z\in\CC_+$
for any $f\in[0,1]$ and they satisfy $|\bd^{(k+1)}(f,z)| \leq \bs c/\Im z$ and
$|\tbd^{(k+1)}(f,z)| \leq 1/\Im z$. By reproducing the
inequalities~\eqref{Imphi}, we also show that $\Im\bd^{(k+1)} > 0$. The other
inequalities are proven similarly, which establishes Item~\emph{i)}. \\

We now show \emph{ii)}. Writing
$\Delta^{(k)}(f,z) = \bd^{(k)}(f,z) - \bd^{(k-1)}(f,z)$ and
$\tilde\Delta^{(k)}(f,z) = \tbd^{(k)}(f,z) - \tbd^{(k-1)}(f,z)$, we get by a
derivation similar to the one made in Section~\ref{unique} that
\begin{align*}
\Delta^{k+1}(f,z) &= \int_0^1 ( {\cal K}_{11}^{(k)}(f,u,z) \Delta^{(k)}(u,z)
        + z {\cal K}_{12}^{(k)}(f,u,z) \tilde\Delta^k(u,z) ) \, du  \\
\tilde\Delta^{k+1}(f,z) &=
\int_0^1 ( z^{-1}{\cal K}_{21}^{(k)}(f,u,z) \Delta^{(k)}(u,z)
  + {\cal K}_{22}^{(k)}(f,u,z) \tilde\Delta^{(k)}(u,z) ) \, du
\end{align*}
where the $\cal K_{ij}^{(k)}$ have the same expressions as the $\cal K_{ij}$
in~\eqref{Kij} except that the $\SS, \ldots$ there are replaced with
$\SS^{(k)}, \ldots$ and the $\SS', \ldots$ are replaced with
$\SS^{(k-1)}, \ldots$. Using the bounds established in~\emph{i)}, we readily
obtain
\[
\begin{array}{ccc}
|\cal K_{11}^{(k)}(f,u,z)| &\leq&
      \bs c \bs\sigma^2 \bs a^2 |z|^2 \bs\gamma(u-f) / (\Im z)^4 , \\
|z \cal K_{12}^{(k)}(f,u,z) | &\leq&
                    \bs c \bs\sigma^2 |z| \bs\gamma(f-u) / (\Im z)^2 , \\
|z^{-1} \cal K_{21}^{(k)}(f,u,z) | &\leq&
                          \bs\sigma^2 |z| \bs\gamma(u-f) / (\Im z)^2 , \\
|\cal K_{22}^{(k)}(f,u,z)| &\leq&
          \bs\sigma^2 \bs a^2 |z|^2 \bs\gamma(f-u) / (\Im z)^4 .
\end{array}
\]
Item~\emph{ii)} follows from these inequalities after integrating over $u$. \\
To show \emph{iii)}, we start by writing
$\bd^{(k)}(f) - \bd^{(k)}(f') = T^{-1} \tr \SS^{(k)}(f)
( \SS^{(k)}(f')^{-1} - \SS^{(k)}(f)^{-1})  \SS^{(k)}(f')$ where we omit the
parameter $z$. Developing the right hand side, we get
\begin{align*}
| \tbz^{(k)}(f) - \tbz^{(k)}(f') |
&= \sigma^2 \Bigl|
   \int_0^1 (\bs\gamma(f-u) - \bs\gamma(f'-u)) \tbd^{(k)}(u) \, du \Bigr| \\
&\leq \frac{\sigma^2}{\Im z} \int_0^1 | \bs\gamma(f-f'+u) - \bs\gamma(u) | \, du
\end{align*}
and similarly for $| \bz^{(k)}(f) - \bz^{(k)}(f') |$. Furthermore,
\begin{multline*}
\frac{({\bs A}{\bs A}^*)(f)}{1+\bz^{(k)}(f)} -
    \frac{({\bs A}{\bs A}^*)(f')}{1+\bz^{(k)}(f')} \\
= \frac{({\bs A}{\bs A}^*)(f) - ({\bs A}{\bs A}^*)(f')}{1+\bz^{(k)}(f)} -
    \frac{(\bz^{(k)}(f) - \bz^{(k)}(f'))({\bs A}{\bs A}^*)(f')}
    {(1+\bz^{(k)}(f)) (1+\bz^{(k)}(f'))} .
\end{multline*}
Recalling that ${\bs A}(f)$ is a trigonometric matrix polynomial, it can be
checked that $\| ({\bs A}{\bs A}^*)(f) - ({\bs A}{\bs A}^*)(f')
\| \leq 4\pi L \bs a^2 |f - f'|$. Item~\emph{iii)} is then obtained by a small
calculation.
\end{proof}

For any $z\in\CC_+$, the map $h_z$ defined before the statement of
Lemma~\ref{recursion} is a map on the Banach space of the continuous
$[0,1] \to\CC_+^2$ functions endowed with the maximum of the supremum
norms of the two components.
Showing the existence of the solution amounts to showing the existence of
a fixed point of $h_z$ that satisfies the inequalities provided in the statement
of Theorem~\ref{det-eq}. By Lemma~\ref{recursion}-\emph{ii)}, the map $h_z$
is a contraction for any $z\in\cal R$. Therefore,
$(\bd^{(k)}(\cdot, z), \tbd^{(k)}(\cdot, z))$ converges for any $z\in\cal R$,
and the limit that we denote $(\bd(\cdot, z), \tbd(\cdot, z))$ is a fixed
point of $h_z$. \\

We now show that from every sequence of integers, one can extract a
subsequence $v(k)$
such that $\bd^{(v(k))}(f,z)$ and $\tbd^{(v(k))}(f,z)$ converge uniformly
on the compact subsets of $[0,1] \times \CC_+$. \\
To that end, we start by showing that on every compact set $K \subset \CC_+$,
the families $\bd^{(k)}(f,z)$ and $\tbd^{(k)}(f,z)$ are equicontinuous on
$[0,1] \times K$. Let $2d > 0$ be the distance from $K$ to $\RR$,
and let $z,z'\in K$ be such that $|z - z'| \leq d/2$. Denote by $\cal C$
the positively oriented circle with center $z$ and radius $d$. Since
$\bd^{(k)}(f,z)$ is holomorphic, we have by Cauchy's formula
\[
\bd^{(k)}(f,z) - \bd^{(k)}(f,z') = \frac{z-z'}{2\imath\pi}
\oint_{\cal C} \frac{\bd^{(k)}(f,w)}{(w-z)(w-z')} \, dw
\]
which shows that
\begin{equation}
\label{equi-z}
| \bd^{(k)}(f,z) - \bd^{(k)}(f,z') | \leq \frac{2\bs c}{d^2} | z - z'| .
\end{equation}
Given two couples $(f,z), (f',z') \in [0,1]\times K$, we have
$|\bd^{(k)}(f,z) - \bd^{(k)}(f',z') | \leq |\bd^{(k)}(f,z) - \bd^{(k)}(f',z) |
+ |\bd^{(k)}(f',z) - \bd^{(k)}(f',z') |$.
In conjunction with Inequality~\eqref{equi-z},
Lemma~\ref{recursion}-\emph{iii)} shows then that the families
$\{ \bd^{(k)}(f,z) \}_{k\in\NN}$ and $\{\tbd^{(k)}(f,z) \}_{k\in\NN}$ are
equicontinuous on the compact set $[0,1] \times K$ (note that the
integral at the right hand side of the inequality in
Lemma~\ref{recursion}-\emph{iii)} converges to zero as $f-f' \to 0$ by the
dominated convergence theorem). These families are moreover bounded on
$[0,1] \times K$. Therefore,
by the Arzel\`a-Ascoli theorem, one can extract from every sequence of
integers a subsequence $p(k)$ such that $\bd^{(p(k))}(f,z)$ and
$\tbd^{(p(k))}(f,z)$ converge uniformly on $[0,1] \times K$. Considering a
sequence of compact subsets of $\CC_+$ who is increasing with respect to
the inclusion and whose union is $\CC_+$, one can establish the existence of
the sequence $v(k)$ by the diagonal process.

Note that the respective limit functions $\bd_v(f,z)$ and $\tbd_v(f,z)$
of $\bd^{(v(k))}(f,z)$ and $\tbd^{(v(k))}(f,z)$ are continuous in the
variable $f$ and holomorphic in the variable $z$. They also satisfy
$|\bd_v(f,z)|\leq \bs c/\Im z$, $|\tbd_v(f,z)|\leq 1/\Im z$, and
$\Im \bd_v$, $\Im \tbd_v$, $\Im (z \bd_v)$, and $\Im (z \tbd_v)$ are
nonnegative by Lemma~\ref{recursion}-\emph{i)} and a passage to the limit.
Furthermore, $h_z(\bd_v(f,z),\tbd_v(f,z))$ exists and is holomorphic on
$\CC_+$. \\
For $z\in\cal R$, $(\bd_v(f,z),\tbd_v(f,z))$ clearly coincides with
$(\bd(f,z), \tbd(f,z))$. Thus
$(\bd_v(f,z),\tbd_v(f,z)) - h_z(\bd_v(f,z),\tbd_v(f,z)) = 0$
on $\cal R$, hence on all $\CC_+$ by analyticity. \\

We just showed that $(\bd^{(k)},\tbd^{(k)})$ converges to the unique solution
$(\bd, \tbd)$ of the system in the statement of Theorem~\ref{det-eq}.
The functions $\bd$ and $\tbd$ are analytical in the variable $z$ and they
satisfy $|\bd(f,z)|\leq \bs c/\Im z$ and $|\tbd(f,z)|\leq 1/\Im z$ along
with the non-negativity conditions.
Hence, for any $f\in[0,1]$, $(T/N) \bd(f,z)$ and $\tbd(f,z)$ are the
Stieltjes transforms of finite positive measures carried by $[0,\infty)$.
By Lemma~\ref{bornes}, $\Im \bd$, $\Im \tbd$, $\Im (z \bd)$, and
$\Im (z \tbd)$ are positive. It remains to show that these measures are
probability measures. We have
\[
- \frac TN \imath y\bd(f,\imath y) =
\frac 1N\sum_{i=1}^N
\frac{-\imath y}
{-\imath y( 1 + \bs\sigma^2\bs\gamma(f)\star\tbd(f,\imath y))
+ \frac{\lambda_i}{1 + \bs\sigma^2\bs\gamma(-f)\star\bd(f,\imath y)}}  .
\]
For any $i$, the denominator behaves as $-\imath y$ when $y\to\infty$.
Hence $-(T/N) \imath y\bd(f,\imath y) \to 1$ as $y\to\infty$, which
establishes the result for $(T/N) \bd$. The result for $\tbd$ is proven
similarly.

\section{Proofs of Theorems~\ref{approx} and~\ref{capa}}
\label{prf-approx}

We start with some notations. Recall that $Q_T(z) = (H_TH_T^* - z I)$.
Denote by $\tQ(z) = (H_T^* H_T - z I)^{-1}$ the resolvent of the self-adjoint
operator $H_T^* H_T$.
The complex number $z\in\CC_+$ will be always written as $z = x+\imath y$.
When there is no ambiguity, the parameter $z$ or the index $T$ will be omitted
for notational simplicity.
We shall denote as $B_{m,n}(k,\ell)$ or $[B]_{m,n}(k,\ell)$ the $(m,n)$
element of the $(k,l)$ block of the matrix representation of $B$, the size of
the blocks being clear from the context. The block itself will be denoted
as $B(k,\ell)$ or $[B](k,\ell)$.

\subsection{Identities, inequalities and basic tools}
\label{basic}

We start with a some technical results.


\begin{lemma}
\label{HQ}
The operator $Q(z) H$ can be continuously extended to a bounded operator
on $l^2(\ZZ)$ (since the domain of $H$ is only dense in $l^2(\ZZ)$) with
adjoint $H^* Q(\bar z)$. The norm of this operator satisfies
$\| Q(z) H \|^2 \leq (y+|z|)/y^2$.
\end{lemma}
\begin{proof}
Given any $a \in \dom(H)$ and any $b \in l^2(\ZZ)$, we have
$\langle Q(z)Ha, b \rangle =  \langle Ha, Q(\bar z) b \rangle =
\langle a, H^* Q(\bar z) b \rangle$, the second inequality being due to the
fact that $Q(\bar z) b \in \dom(HH^*) \subset \dom(H^*)$. Therefore,
$H^* Q(\bar z)$ is the adjoint of $Q(z)H$, and since it is defined on all
$l^2(\ZZ)$, the operator $Q(z)H$ can be extended to a bounded operator.
We furthermore have
\[
\| Q(z)H H^* Q(\bar z) \| = \| Q(z)(I + \bar z Q(\bar z)) \|
\leq \frac{y+|z|}{y^2}
\]
hence the last result.
\end{proof}

The following bounds will also be often used in the proof without mention:
\begin{gather*}
| Q_{n,n'}(k,\ell)(z) | \leq \| Q(k,\ell)(z) \| \leq
  \| Q(z) \| \leq \frac{1}{y}, \\
\sum_{k\in\ZZ} \sum_{n=0}^{N-1} | Q_{n,p}(k,\ell) |^2 =
[Q^* Q]_{p,p}(\ell,\ell) \leq \frac{1}{y^2}  \\
\sum_{\ell\in\ZZ} \sum_{p=0}^{N-1} | [QH]_{n,p}(m,\ell) |^2
= [ Q HH^* Q^*]_{p,p}(m,m) \leq \frac{y+|z|}{y^2}
\end{gather*}

\begin{lemma}
\label{HQH}
$H^* Q(z) H = I + z \tQ(z)$ and $H \tQ H^* = I + z Q(z)$ by continuous
extensions of the left hand members.
\end{lemma}
\begin{proof}
The proof is based on the polar decomposition $H = U | H|$ of
$H$~\cite[\S VI.2.7]{kat-livre80}.
We recall that $|H| = (H^*H)^{1/2}$, that $U$ is a partial isometry from
$\overline{\ran(|H|)}$ to $\overline{\ran(H)}$ (here $\ran(\cdot)$ is the
range of an operator and $\overline{\ran(\cdot)}$ the closure of this range),
and that $H^* = |H| U^*$. Letting $a\in\dom(H)$ and $b = Q(z) H a$, we have
$Ha = HH^* b - z b$, or equivalently $b = z^{-1} H (H^*b - a)$ that we rewrite
as $b = U h$ with $h = z^{-1} |H| ( |H| U^* b - a )$. Since
$h \in \ran(|H|)$, $U^* U h = h$, therefore $h = z^{-1} |H| ( |H| h - a )$
or equivalently $h = (|H|^2 - zI)^{-1} |H| a = |H| (|H|^2 - zI)^{-1} a$.
Finally, the vector $w = H^* Q(z) H a = H^* b$ satisfies
$w = |H| U^* U h = |H|^2 (|H|^2 - zI)^{-1} a =
H^* H \tQ(z) a = (H^* H - z)\tQ(z)a +z\tQ(z) a = a + z\tQ(z)a$.
\end{proof}

The following lemma can be shown by direct calculation.

\begin{lemma}[Differentiation formulas]
\label{deriv}
Given an integer $M\in\NN$, let $H^M$ be the random operator represented by
the matrix $H^M = [ \1_{|k|\leq M} \1_{|\ell|\leq M} H(k,\ell) ]$
and let $Q^M(z) = (H^M H^{M*} - zI)^{-1} =
[ [Q^M_{p,q}(k,\ell)(z)]_{p,q=0}^{N-1} ]_{k,\ell\in\ZZ}$ be the resolvent
of the self adjoint operator $H^M H^{M*}$.
For any integers $i,j,k,l, n,t,p,q$ satisfying
$-M \leq i, j, k,\ell \leq M$, $|i-j|\leq L$,
$0\leq n,t,p\leq N-1$, and $0\leq q\leq T-1$, it holds that
\begin{align*}
\frac{\partial Q^M_{n,t}(k,\ell)}{\partial X_{p,q}(i,j)} &=
- Q_{n,p}^M(k,i) [ H^{M*} Q^M]_{q,t}(j,\ell) \quad \text{and} \\
\frac{\partial Q_{n,t}^M(k,\ell)}{\partial \bar X_{p,q}(i,j)} &=
- [ Q^M H^M ]_{n,q}(k,j) Q^M_{p,t}(i,\ell) .
\end{align*}
\end{lemma}

The following basic property of convolution operators is well
known~\cite{bot-sil-toep}.
\begin{proposition}
\label{laurent}
A convolution operator $B = [ B_{k-\ell} ]_{k,\ell\in\ZZ}$ where the blocks
are $N\times N$ matrices is a bounded operator on $l^2(\ZZ)$ if and only if
there exists a bounded $N\times N$ matrix function $\bs B(f)$ on $[0,1]$ such
that
\[
B_k = \int_0^1 \exp(-2\imath\pi k f) \bs B(f) \, df \quad k \in \ZZ ,
\]
\emph{i.e.}, the $B_k$ are the Fourier coefficients of the $1$-periodic
function equal to $\bs B(f)$ on $[0,1]$.
\end{proposition}

Let us identify the function $\bs B(f)$ on $[0,1]$ with the multiplication
operator $\bs B : {\cal L}^2([0,1]\to\CC^N) \to {\cal L}^2([0,1]\to\CC^N)$
who sends $\bs g(f)$ to $(\bs B \bs g)(f) = \bs B(f) \bs g(f)$. Clearly,
the spectral norm of this operator coincides with
$\sup_{f\in[0,1]} \| \bs B(f) \|$.
Let
\[
{\cal F}_N : {\cal L}^2([0,1]\to\CC^N) \to l^2(\ZZ), \quad
{\cal F}_N(\bs g) =
  \Bigl( \int_0^1 \exp(-2\imath\pi k f) \bs g(f) \, df\Bigr)_{k\in\ZZ}
 = (g_k)_{k\in\ZZ}
\]
be the isometric operator that sends $\bs g$ to the sequence of its
Fourier coefficients. Then it holds that $B = {\cal F}_N \bs B {\cal F}_N^*$.
\\

We now introduce the two basic tools used in the proof of Theorem~\ref{approx}.

\begin{proposition}
\label{ip}
Let $\xi \in \CC^M$ be a complex Gaussian centered vector such that
$\EE \xi\xi^T = 0$ and $\EE \xi\xi^* = \Xi$. Let $\Gamma\, : \CC^M \to \RR$
be a $C^1$ function polynomially bounded together with its derivatives
(when seeing $\Gamma$ as a function from $\RR^{2M}$ to $\RR$). Then
\[
\EE[ \xi_i \Gamma(\xi) ] = \sum_{m=1}^M [\Xi]_{im}
\EE\left[ \frac{\partial \Gamma(\xi)}{\partial \bar\xi_m} \right]
\]
\end{proposition}
This proposition can be proven by an integration by parts.
We shall call it for this reason the IP formula.

\begin{proposition}
\label{np}
Let $\xi \in \CC^M$ and $\Gamma\, : \CC^M \to \RR$ be as in the statement
of Proposition~\ref{ip}, and let
$\nabla_z\Gamma = [ \partial\Gamma / \partial z_1, \ldots,
\partial\Gamma / \partial z_M ]^T$ and
$\nabla_{\overline{z}}\Gamma = [ \partial\Gamma / \partial \overline{z_1},
\ldots, \partial\Gamma / \partial \overline{z_M} ]^T$.
Then the following inequality holds true:
\[
\var {\Gamma}({\xi}) \leq
 \EE \bigl[ \nabla_z \Gamma({\xi})^T \ {\Xi} \
 \overline{\nabla_z \Gamma({\xi})} \bigr]
 +
 \EE \left[ \left( \nabla_{\overline{z}}  \Gamma({\xi}) \right)^*
 \ {\Xi} \
 \nabla_{\overline{z}} \Gamma({\xi})
 \right]
 \ .
\]
\end{proposition}
This inequality is known as the Poincar\'e-Nash (PN) inequality. For a
proof, see~\cite[Prop.~2.1.6]{pas-livre}.

\subsection{System of equations controlling $\EE Q(k, k+m)$ and
$\EE \tQ(k, k+m)$.}
\label{core-prf}

We now enter the core of the proof, which consists in showing that
the $\EE Q(k,\ell)$ and the $\EE\tQ(k, \ell)$ satisfy a perturbed infinite
system of equations. Our derivations will be greatly simplified by the fact
that due to the ergodicity of $H_T$, the operators $\EE Q(z)$ and
$\EE\tQ(z)$ are bounded convolution operators for any $z\in\CC_+$.

We start with some variance controls.
The following lemma is proven in Appendix~\ref{anx-varQ}. The proof is based
on the use of Proposition~\ref{np} along with an operator truncation argument.

\begin{lemma}
\label{varQ}
For any $k,\ell \in\ZZ$ and any $n \in \{0,\ldots,N-1\}$ and
$t \in \{0,\ldots, T-1 \}$, we have
\begin{gather*}
\var Q_{p,q}(k,\ell) \leq
\frac{2}{T}\frac{\bs\sigma^2 \bs g (|z|+y)}{y^4},
\quad
\var \tr Q(k,\ell) \leq \frac{2 \bs c \bs\sigma^2 \bs g (y+|z|)}{y^4} , \\
\text{and} \
\var [H^* Q]_{p,q}(k,\ell) \leq
\frac{2}{T}\frac{\bs\sigma^2 \bs g (|z|+y)^2}{y^4} .
\end{gather*}
\end{lemma}

Our purpose is now to find an equation satisfied by $\EE Q_{p,q}(k,\ell)$. To
that end, we start by writing $\EE[ [HH^* Q]_{p,q}(k,\ell) ] =
\chi_1 + \chi_2 + \chi_3 + \chi_4$ where $\chi_1$, $\chi_2$ and $\chi_3$
are given by the finite sums
\begin{align*}
\chi_1 &= \sum_{i,j} \sum_{n,t = 0}^{N-1,T-1}
\EE[ X_{p,t}(k,i) \bar X_{n,t}(j,i) Q_{nq}(j,\ell) ] , \\
\chi_2 &=
\sum_{i,j} \sum_{n,t = 0}^{N-1,T-1}
\EE[ X_{p,t}(k,i) \bar A_{n,t}(j,i) Q_{nq}(j,\ell) ] , \\
\chi_3 &=
\sum_{i,j} \sum_{n,t = 0}^{N-1,T-1}
\EE[ A_{p,t}(k,i) \bar X_{n,t}(j,i) Q_{nq}(j,\ell) ] ,
\end{align*}
and where
\[
\chi_4 = \EE[ [AA^* Q]_{p,q}(k,\ell) ] .
\]
Starting with $\chi_1$, the idea is to develop
$\EE[ X_{p,t}(k,i) \bar X_{n,t}(j,i) Q_{nq}(j,\ell) ]$ with the help of the
IP formula. Thanks to this formula, we expect that
\begin{align}
& \EE[ X_{p,t}(k,i) \bar X_{n,t}(j,i) Q_{nq}(j,\ell) ] \nonumber \\
& = \sum_{u,v}
\EE[ X_{p,t}(k,i) \bar X_{p,t}(u,v) ]
\EE\Bigl[ \frac{\partial (\bar X_{n,t}(j,i) Q_{nq}(j,\ell))}
{\partial \bar X_{p,t}(u,v)} \Bigr]  \nonumber \\
&=  \sum_{u,v} \frac{\gamma_T(k-u) \phi_T(k-i)^2}{T} \delta_{k-i , u-v}
\Bigl\{
\delta_{p,n}\delta_{u,j}\delta_{v,i} \EE[ Q_{nq}(j,\ell) ] \nonumber \\
&
\ \ \ \ \ \ \ \ \ \ \ \ \ \ \ \ \ \ \ \
\ \ \ \ \ \ \ \ \ \ \ \ \ \ \ \ \ \ \ \ \ \ \ \ \ \ \ \
- \EE[ \bar X_{n,t}(j,i) [QH]_{n,t}(j,v) Q_{p,q}(u,\ell) ]
\Bigr\} \nonumber \\
&= \sum_r \frac{\gamma_T(r) \phi_T(k-i)^2}{T} \Bigl\{
\delta_{p,n}\delta_{k,j}\delta_{r,0} \EE[ Q_{nq}(j,\ell) ] \nonumber \\
&
\ \ \ \ \ \ \ \ \ \ \ \ \ \ \ \ \ \ \ \
\ \ \ \ \ \ \ \ \ \ \ \ \ \ \ \ \ \ \ \
- \EE[ \bar X_{n,t}(j,i) [QH]_{n,t}(j,i-r) Q_{p,q}(k-r,\ell) ]
\Bigr\} .
\label{ip-chi1}
\end{align}

Note that the vector $\xi$ in the statement of Proposition~\ref{ip} is finite
dimensional. Therefore, Equation~\eqref{ip-chi1} needs to be justified.
Given an integer $M\in\NN$, let $H^M$ and $Q^M$ be as in the statement of
Lemma~\ref{deriv}.
When $M \geq |k| \vee |i| \vee |j|$, we get
by applying Proposition~\ref{ip} and repeating the derivation above that
\begin{align}
& \EE[ X_{p,t}(k,i) \bar X_{n,t}(j,i) Q_{nq}^M(j,\ell) ] \nonumber \\
&= \sum_{r=k\vee i - M}^{k\wedge i + M}
\frac{\gamma_T(r) \phi_T(k-i)^2}{T} \Bigl\{
\delta_{p,n}\delta_{k,j}\delta_{r,0} \EE[ Q_{nq}^M(j,\ell) ] \nonumber \\
&
\ \ \ \ \ \ \ \ \ \ \ \ \ \ \ \ \ \ \ \
\ \ \ \ \ \ \ \ \ \ \ \ \
- \EE[ \bar X_{n,t}(j,i) [Q^MH^M]_{n,t}(j,i-r) Q_{p,q}^M(k-r,\ell) ]
\Bigr\} .
\label{rhs-M}
\end{align}
The argument provided in the proof of Lemma~\ref{varQ} shows that
$Q_{nq}^M(j,\ell) \xrightarrow{\text{a.s.}} Q_{nq}(j,\ell)$ as $M\to\infty$.
Since $| X_{p,t}(k,i) \bar X_{n,t}(j,i) Q_{nq}^M(j,\ell) |
\leq | X_{p,t}(k,i) \bar X_{n,t}(j,i) | / y$, it holds by the dominated
convergence theorem that
$\EE[ X_{p,t}(k,i) \bar X_{n,t}(j,i) Q_{nq}^M(j,\ell) ]
\to \EE[ X_{p,t}(k,i) \bar X_{n,t}(j,i) Q_{nq}(j,\ell) ]$ as $M\to\infty$.
Using Lemma~\ref{HQ} and the summability of $\gamma$, we can show by a
similar argument that that the right hand side of~\eqref{rhs-M} converges to
the right hand side of~\eqref{ip-chi1}. Equation~\eqref{ip-chi1} is shown. \\
Turning to $\chi_2$, we obtain by a similar derivation
\[
\chi_2 =
- \sum_{i,j,r} \sum_{n,t=0}^{N-1,T-1}
\frac{\gamma_T(r) \phi_T(k-i)^2}{T}
\EE[ \bar A_{n,t}(j,i) [QH]_{n,t}(j,i-r) Q_{p,q}(k-r,\ell) ]  .
\]
Considering $\chi_1 + \chi_2$, and taking the sum over $j$ and $n$, we get
\begin{align*}
\chi_1 + \chi_2 &= \sum_{i,r} \sum_{t=0}^{T-1}
\frac{\gamma_T(r) \phi_T(k-i)^2}{T}
\Bigl\{ \delta_{r,0} \EE[Q_{p,q}(k,\ell)] - \\
&
\ \ \ \ \ \ \ \ \ \ \ \ \ \ \ \ \ \ \ \
\ \ \ \ \ \ \ \ \ \ \ \ \ \ \ \ \ \ \ \
\EE[ [H^*QH]_{t,t}(i,i-r) Q_{p,q}(k-r,\ell) ] \Bigr\} \\
&= \sum_{r,s} \sum_{t=0}^{T-1} \frac{\gamma_T(r) \phi_T(s)^2}{T}
\EE\Bigl[ Q_{p,q}(k-r,\ell) \times \\
&
\ \ \ \ \ \ \ \ \ \ \ \ \ \ \ \ \ \ \ \
( I_{t,t}(k-s,k-s-r) - [H^*QH]_{t,t}(k-s,k-s-r) ) \Bigr] \\
&= - z \sum_{r,s} \gamma_T(r) \phi_T(s)^2
   \EE\Bigl[ Q_{pq}(k-r,\ell) \frac{\tr \tQ(k-s,k-s-r)}{T} \Bigr]
\end{align*}
where the third equality is due to Lemma~\ref{HQH}. In the second equation,
we considered that $\delta_{r,0} = I_{t,t}(k-s,k-s-r)$, the $(t,t)$ element of the
$(k-s, k-s-r)$ block of the ``$N\times N$ block-matrix representation'' of the
identity operator. Turning to $\chi_3$, we have
\begin{align*}
\chi_3 &= \sum_{i,j,u,v} \sum_{n,t} A_{p,t}(k,i)
\EE[ X_{n,t}(u,v) \bar X_{n,t}(j,i) ]
\EE\Bigl[ \frac{\partial Q_{n,q}(j,\ell)}{\partial X_{n,t}(u,v)}\Bigr] \\
&= -\sum_{i,j,u,v} \sum_{n,t} A_{p,t}(k,i)
 \frac{\gamma_T(u-j) \phi_T(j-i)^2}{T} \delta_{j-i , u-v}
\EE[ Q_{n,n}(j,u) [ H^* Q ]_{t,q}(v,\ell) ] \\
&= -\sum_{i,r,s} \sum_{t} \gamma_T(r) \phi_T(s)^2 A_{p,t}(k,i)
\EE\Bigl[ \frac{\tr Q(s+i,s+i+r)}{T} [ H^* Q ]_{t,q}(i+r,\ell) \Bigr] .
\end{align*}
We now ``decouple'' the terms within the expectations in the expressions
of $\chi_1+\chi_2$ and $\chi_3$ by using the inequality
$| \EE XY - \EE X \EE Y | = | \EE[(X-\EE X)(Y-\EE Y)] | \leq
(\var X)^{1/2} (\var Y)^{1/2}$ along with Lemma~\ref{varQ}.
By the ergodicity of $H$, terms such as $\EE[Q_{p,q}(k,\ell)]$,
$\EE[[H^*Q]_{p,q}(k,\ell)]$, or $\EE[[AH^*Q]_{p,q}(k,\ell)]$ depend on
$k-\ell$ only. With a small notation abuse, we shall henceforth denote
the first of these terms as $\EE[Q_{p,q}(k,\ell)]$ or $\EE[Q_{p,q}(k-\ell)]$
interchangeably, and similarly for the other terms. With these notations, we
get
\[
\chi_1 + \chi_2 = - z \sigma_T^2 \sum_{r} \gamma_T(r)
\EE[ Q_{pq}(k-r-\ell) ] \Bigl[ \frac{\tr \EE\tQ(r)}{T} \Bigr]
+ \varepsilon_{p,q}(k-\ell)
\]
where
\[
| \varepsilon_{p,q}(k-\ell)  | \leq
\frac{2\sqrt{\bs c}\bs\sigma^4\bs g^2}{T^{3/2}} \frac{|z|(|z|+y)}{y^4} ,
\]
and
\[
\chi_3 = - \sigma_T^2
\sum_{i,r} \sum_{t} \gamma_T(r) A_{p,t}(k,i)
\Bigl[ \frac{\tr \EE Q(-r)}{T} \Bigr]
\EE\bigl[[ H^* Q ]_{t,q}(i,\ell-r) \bigr] + \varepsilon'_{p,q}(k-\ell)
\]
where $\varepsilon'$ satisfies by Lemma~\ref{varQ}
\[
|\varepsilon_{p,q}(k-\ell)|' \leq
\frac{2\sqrt{\bs c} \bs\sigma^4 \bs g^2 (y+|z|)^{3/2}}{y^4}
\frac{1}{T^{3/2}} \sum_{i,t} | A_{p,t}(k,i) |
\leq
\frac{2\sqrt{\bs c} \bs\sigma^4 \bs g^2 \bs a (y+|z|)^{3/2}}{y^4}
\frac{1}{T}
\]
since $\sum_{i,t} |A_{p,t}(k,i)| \leq \sqrt{T} \sum_i \| A(k,i) \|
\leq \sqrt{T} \bs a$. Recalling that $A$ is a convolution operator, we have
\[
\chi_3 = - \sigma_T^2
\sum_{r} \gamma_T(r) \Bigl[ \frac{\tr \EE Q(-r)}{T} \Bigr]
\EE\bigl[[ A H^* Q ]_{p,q}(k+r-\ell) \bigr] + \varepsilon'_{p,q}(k-\ell) .
\]
By the identity $HH^*Q(z) = (HH^* - zI) Q(z) + zQ = I + z Q(z)$, we get
\begin{align*}
\EE Q_{pq}(k-\ell) &= -z^{-1} I_{p,q}^N(k-\ell) + z^{-1}(\chi_1+\chi_2)
                      + z^{-1} \EE[ AH^* Q]_{p,q}(k-\ell) , \\
\EE[ AH^* Q]_{p,q}(k-\ell) &= \chi_3 + \chi_4 .
\end{align*}
Specifically,
\begin{align*}
\EE Q_{p,q}(k) &= -z^{-1} I_{p,q}^N(k)
- \sigma_T^2 \sum_{r} \gamma_T(r) \Bigl[ \frac{\tr \EE\tQ(r)}{T} \Bigr]
\EE Q_{pq}(k-r) \\
&\phantom{=} + z^{-1} \EE[ AH^* Q]_{p,q}(k) + \varepsilon_{p,q}(k) , \\
\EE[ A H^* Q ]_{p,q}(k) &=
- \sigma_T^2 \sum_{r} \gamma_T(-r) \Bigl[ \frac{\tr \EE Q(r)}{T} \Bigr]
   \EE[ A H^* Q ]_{p,q}(k-r) \\
&\phantom{=}
+ \EE [AA^* Q]_{p,q}(k)  + \varepsilon'_{p,q}(k) .
\end{align*}
Similar derivations lead to the identities
\begin{align*}
\EE \tQ_{p,q}(k) &= -z^{-1} I_{p,q}^T(k) -\sigma_T^2 \sum_r
\gamma_T(-r) \Bigl[ \frac{\tr \EE Q(r)}{T} \Bigr] \EE\tQ_{p,q}(k-r) \\
&\phantom{=}
+ z^{-1} \EE[A^*H\tQ]_{p,q}(k) + \tilde\varepsilon_{p,q}(k) , \\
\EE[ A^* H \tQ ]_{p,q}(k) &=
-\sigma_T^2 \sum_r \gamma_T(r) \Bigl[ \frac{\tr \EE\tQ(r)}{T} \Bigr]
\EE[ A^* H \tQ ]_{p,q}(k-r) \\
&\phantom{=} + \EE[ A^* A \tQ ]_{p,q}(k) + \tilde\varepsilon'_{p,q}(k)
\end{align*}
where
\[
| \tilde\varepsilon_{p,q}(k)  | \leq
\frac{2\bs\sigma^4\bs g^2|z|(|z|+y)}{y^4} \frac{1}{T^{3/2}}
\quad \text{and} \quad
|\tilde\varepsilon_{p,q}(k)|' \leq
\frac{2\bs\sigma^4 \bs g^2 \bs a (y+|z|)^{3/2}}{y^4}
\frac{1}{T} .
\]

\subsection{Passage to the frequency domain; End of the proof}
\label{t->f}

For any $z\in\CC_+$, we can identify the matrix functions $\SS_T(\cdot,z)$ and
$\tSS_T(\cdot,z)$ defined in the statement of Theorem~\ref{det-eq} with the
multiplication operators
$(\SS_T(\cdot,z) \bs g)(f) = \SS_T(f,z) \bs g(f)$ and
$(\tSS_T(\cdot,z) \tilde{\bs g})(f) = \tSS_T(f,z) \tilde{\bs g}(f)$ on
${\mathcal L}^2([0,1]\to\CC^N)$ and ${\mathcal L}^2([0,1]\to\CC^T)$
respectively.
Given a function $\bs g \in {\mathcal L}^2([0,1]\to\CC^N)$ such that
$\| \bs g \| > 0$, one can show by derivations similar to those of
Section~\ref{prf-det-eq} that
\[
\langle (\Im \SS(\cdot,z)) \bs g, \bs g \rangle =
\int_0^1
\frac{{\bs g}(f)^* \SS(f) ( \SS(f)^{-*} - \SS(f)^{-1} )
\SS(f)^* \bs g(f)}{2\imath} \, df  > 0
\]
for any $z\in\CC_+$, where $\Im \SS = (\SS - \SS^*)/(2\imath)$. The
self-adjoint operator $\Im \SS(\cdot,z)$ is therefore positive for any
$z\in \CC_+$. Similarly, one
can show that $\Im (z\SS(\cdot,z))$ is positive on $\CC_+$, and so is the
case for $\Im \tSS(\cdot,z)$ and $\Im (z\tSS(\cdot,z))$.

Let $S_T(z) = [ S_T(k-\ell)(z) ]_{k,\ell\in\ZZ}$ and
$\tS_T(z) = [ \tS_T(k-\ell)(z) ]_{k,\ell\in\ZZ}$ the convolution operators
obtained through the isometries
\[
S_T(z) = {\mathcal F}_N \SS_T(\cdot,z) {\mathcal F}_N^{^*}
\quad \text{and} \quad
\widetilde S_T(z) = {\mathcal F}_T \tSS_T(\cdot,z) {\mathcal F}_T^{^*} .
\]
Then we have the following result:

\begin{lemma}
For any $z\in\CC_+$, the matrix blocks of the operators $S_T$ and $\tS_T$
satisfy the equations (where the parameter $z$ is omitted)
\begin{align*}
S_T(k) &= - z^{-1} \delta_{k,0} I_N
- \sigma_T^2 \sum_r \gamma_T(r) \td_T(r) S_T(k-r) + z^{-1} P_T(k) , \\
P_T(k) &= - \sigma_T^2 \sum_r \gamma_T(-r) \varphi_T(r) P_T(k-r)
   + [A_TA_T^* S_T](k) , \\
\tS_T(k) &= - z^{-1} \delta_{k,0} I_T
-\sigma_T^2 \sum_r \gamma_T(-r) \varphi_T(r) \tS_T(k-r)
       + z^{-1} \widetilde P_T(k) , \\
\widetilde P_T(k) &=
- \sigma_T^2 \sum_r \gamma_T(r) \td_T(r) \widetilde P_T(k-r)
       + [A^*_TA_T \tS_T](k)
\end{align*}
where
\[
\varphi_T(r) = \frac{\tr S_T(r)}{T}
\quad \text{and} \quad
\td_T(r) = \frac{\tr \tS_T(r)}{T}  .
\]
The convolution operators $P_T = [ P_T(k-\ell)]_{k,\ell\in\ZZ}$ and
$\widetilde P_T = [ \widetilde P_T(k-\ell)]_{k,\ell\in\ZZ}$ satisfy the 
inequality 
$\| P(z) \| \vee \| \widetilde P(z)\| \leq \bs a^2 | z | / (\Im z)^2$. 
In addition, $\| S_T \| \vee \| \tS \| \leq (\Im z)^{-1}$ and furthermore,
$\Im S_T$, $\Im \tS_T$, $\Im(z S_T)$ and $\Im (z \tS_T)$ are positive. \\
There exist four unique bounded convolution operators
$S_T = [ S_T(k-\ell) ]_{k,\ell\in\ZZ}$,
$\tS_T = [ \tS_T(k-\ell) ]_{k,\ell\in\ZZ}$,
$P_T = [ P_T(k-\ell) ]_{k,\ell\in\ZZ}$ and
$\widetilde P_T = [ \widetilde P_T(k-\ell) ]_{k,\ell\in\ZZ}$
satisfying the equations above in association with the positivity constraints
on $\Im S_T$, $\Im \tS_T$, $\Im(z S_T)$ and $\Im (z \tS_T)$.
\end{lemma}
\begin{proof}
Equation~\eqref{eq-S} can be rewritten as
\begin{align*}
-z \SS(f) - z \sigma^2 (\bs\gamma(f)\star \tbd(f)) \SS(f) + \bs P(f) &= I, \\
(1 + \sigma^2 \bs\gamma(-f)\star \bd(f)) \bs P(f)
                                    &= (\bs A\bs A^*)(f) \SS(f) .
\end{align*}
Taking the Fourier transforms of $\SS(f)$ and $\bs P(f)$, we recover the
first two equations in the statement. The other two equations are obtained
similarly. We know that the multiplication operator $\SS$ satisfies
$\| \SS \| \leq 1/ \Im z$, $\Im \SS > 0$ and $\Im (z\SS) > 0$.
From the expression of $\bs P(f)$ we see that the associated multiplication
operator has its norm bounded by $\bs a^2 |z| / (\Im z)^2$.
Since $\cal F_N$ are $\cal F_T$ are isometries, the norms and positivity
constraints are deduced at once. \\
Finally, since a bounded convolution operator is uniquely determined by
its associated multiplication operator, and since the operators
$\SS$ and $\tSS$ are uniquely determined by the conditions of
Theorem~\ref{det-eq}, we obtain the uniqueness alluded to in the last part
of the statement.
\end{proof}

We are now in position to establish the first statement of
Theorem~\ref{approx}. Since
\[
\frac{\tr S_T(0)(z)}{N} = \frac 1N \int_0^1 \tr \bs S_T(f) \, df
= \frac TN \int_0^1 \bs\varphi_T(f,z) \, df ,
\]
what we need to show is that $N^{-1} \tr(\EE Q_T(0)(z) - S_T(0)(z)) \to 0$
for any $z\in\CC_+$. To that end, we start by showing that
\[
\Delta_T = \sup_{k\in\ZZ} \frac{\lf \EE Q(k) - S(k) \rf}{\sqrt{T}}
\quad \text{and} \quad
\widetilde\Delta_T = \sup_{k\in\ZZ}
        \frac{\lf \EE \tQ(k) - \tS(k) \rf}{\sqrt{T}}
\]
converge to zero as $T\to\infty$ in a certain region of $\CC_+$, where
$\lf\cdot\rf$ is the Frobenius norm. Indeed,
we have
\begin{align*}
\EE Q(k) - S(k) &=
- \sigma_T^2 \sum_r \gamma_T(r) \frac{\tr\EE\tQ(r)}{T}
( \EE Q(k-r) - S(k-r) ) \\
&\phantom{=}
- \sigma_T^2 \sum_r \gamma_T(r) \frac{\tr\EE\tQ(r) - \tr\tS(r)}{T} S(k-r) \\
&\phantom{=}
+ z^{-1} ( \EE[AH^*Q](k) - P(k)) + E(k) , \\
\EE[AH^*Q](k) - P(k) &=
- \sigma_T^2 \sum_r \gamma_T(-r) \frac{\tr\EE Q(r)}{T}
                        (\EE[AH^*Q](k-r) - P(k-r) ) \nonumber \\
&\phantom{=}
- \sigma_T^2 \sum_r \gamma_T(-r) \frac{\tr\EE Q(r)-\tr S(r)}{T} P(k-r)
                  \nonumber \\
&\phantom{=}
+ [AA^*(\EE Q - S)](k) + E'(k) \nonumber
\end{align*}
where $\lf E(k) \rf = \Bigl(\sum_{p,q} | \varepsilon_{p,q}(k) |^2\Bigr)^{1/2}
\leq C(z) / \sqrt{T}$ and $\lf E'(k) \rf \leq C'(z)$, with $C(z)$ and $C'(z)$
being bounded on the compact subsets of $\CC_+$. Notice that
\begin{multline}
\Bigl| \frac{\tr[\EE Q - S](r)}{T} \Bigr| \leq
\frac{\sum_{p} | [\EE Q - S]_{p,p}(r) |}{T} \\
\leq \frac{(N\sum_{p} | [\EE Q - S]_{p,p}(r) |^2)^{1/2}}{T} \leq
\sqrt{\bs c} \frac{\lf [\EE Q - S](r) \rf}{\sqrt{T}}
\label{trQ-S}
\end{multline}
and similarly,
$| \tr[\EE \tQ - \tS](r) /T | \leq \lf \EE \tQ - \tS](r) \rf / \sqrt{T}$.
We also have $\lf S(k) \rf^2 / T \leq \bs c \| S(k) \|^2 \leq \bs c / y^2$
and $\lf P(k) \rf^2 / T \leq \bs c \| P(k) \|^2 \leq
\bs c \bs a^4 |z|^2 / y^4$. \\
Writing
\[
\Delta'_T = \sup_k \frac{\lf \EE[AH^*Q](k) - P(k) \rf}{\sqrt{T}} ,
\]
we easily get from the expression of $\EE Q(k) - S(k)$ above
\[
\Delta_T \leq \frac{\bs\sigma^2 \bs g}{y} \Delta_T +
 \frac{\bs\sigma^2 \bs g \sqrt{\bs c}}{y} \widetilde\Delta_T +
 \frac{\Delta'_T}{|z|} + \frac{C(z)}{T} .
\]
Using the norm inequality
$\lf UV \rf \leq \| U \|\,\lf V \rf$ where $U$ and $V$ are two matrices, we
observe that
\begin{align*}
\lf [AA^*(\EE Q - S)](k) \rf &= \lf \sum_{i,j} A(i) A(j)^*
                                        [\EE Q - S](j-i+k) \rf  \\
&\leq \sum_{i,j} \| A(i) \| \, \| A(j) \| \, \lf [\EE Q - S](j-i+k) \rf
\end{align*}
and we get from the expression of $\EE[AH^*Q](k) - P(k)$ that
\[
\Delta'_T \leq
\frac{\bs\sigma^2 \bs g \bs c}{y} \Delta'_T +
\frac{\bs\sigma^2 \bs g \sqrt{\bs c} \bs a^2 |z|}{y^2} \Delta_T +
\bs a^2 \Delta_T + \frac{C'(z)}{\sqrt{T}} ,
\]
and we have two similar inequalities for $\widetilde\Delta_T$ and
$\widetilde\Delta'_T = \sup_k \lf \EE[A^*H\tQ](k) - \widetilde P(k) \rf
/ \sqrt{T}$. These four inequalities can be written in a compact form
as
\[
\bs\Delta_T(z) \leq  \bs M(z) \bs\Delta_T(z) + \frac{1}{\sqrt{T}} \bs E(z)
\]
where $\bs\Delta_T(z) = [\Delta_T, \Delta'_T, \widetilde\Delta_T,
\widetilde\Delta'_T ]^T$, the inequality is element wise, $\bs M(z)$ is a
$4\times 4$ matrix independent on $T$ whose norm converges to zero as
$y\to\infty$, and $\bs E(z)$ is vector whose norm is bounded on any compact
set of $\CC_+$. \\

This result coupled with the inequalities~\eqref{trQ-S} shows that for $y$
large enough, it holds that $N^{-1} \tr(\EE Q_T(0)(z) - S_T(0)(z)) \to 0$. 
Recall that
$N^{-1} \tr\EE Q_T(0)(z)$ and $N^{-1} \tr S_T(0)(z)$ are holomorphic and
bounded on the compact subsets of $\CC-\RR_+$. By the normal family theorem,
from any sequence of integers depending on $T$, we can extract a subsequence
$v(T)$ such that
$N(v(T))^{-1} \tr(\EE Q_{v(T)}(0)(z) - S_{v(T)}(0)(z)) \to 0$ on the
compact subsets of $\CC-\RR_+$. But since the limit is zero for large enough
$y$, it is identically zero on $\CC-\RR_+$. The convergence~\eqref{cvg-st} in
the statement of Theorem~\ref{approx} is proven. To establish the second part
of the statement, we need the tightness of $\bs\pi_T$.

\begin{lemma}
\label{tight}
The sequence $\bs\pi_T$ is tight.
\end{lemma}
\begin{proof}
For any small $\varepsilon > 0$, any $T \in\NN$, any $f\in[0,1]$ and any
eigenvalue $\lambda_{i,T}(f)$ of $(\bs A_T\bs A_T^*)(f)$,
writing $\bz_T(f,\imath y) = \sigma_T^2 \bs\gamma_T(f)\star\tbd_T(f,\imath y)$
and $\tbz_T(f,\imath y) = \sigma_T^2 \bs\gamma_T(-f)\star\bd_T(f,\imath y)$,
we have
\[
\Bigl| \frac{-\imath y}
{-\imath y - \imath y \tbz_T(f,\imath y) +
\frac{\lambda_{i,T}(f)}{1 + \bz_T(f,\imath y)} }
- 1 \Bigr|
=
\Bigl| \frac{\tbz_T(f,\imath y)
- \frac{\lambda_{i,T}(f)}{\imath y + \imath y \bz_T(f,\imath y)}}
{ 1 + \tbz_T(f,\imath y)
- \frac{\lambda_{i,T}(f)}
{\imath y + \imath y \tbz_T(f,\imath y)} } \Bigr|
< \varepsilon
\]
when $y > C_\varepsilon = {2(\bs\sigma^2 \bs g + \bs a^2)}/{\varepsilon}$,
thanks to  the inequalities $|\bz_T(f,\imath y)| \leq \bs\sigma^2\bs g / y$
and $\Im (\imath y \tbz_T(f,\imath y)) > 0$ as $y > 0$.
This implies that the Stieltjes transform
$\bs p_T(z) = N^{-1} \int_0^1 \tr \SS_T(f,z) \, df$ of $\bs\pi_T$ satisfies
$| \imath y \bs p_T(-\imath y) + 1 | < \varepsilon$ when $y > C_\varepsilon$.
Since the family $(m_T(z))_{T\in\NN}$ is bounded on
the compact subsets of $\CC_+$, from any sequence of integers, one can
extract a subsequence $v(T)$ such that $\bs p_{v(T)}(z)$ converges to a
function $\bs p^v(z)$ uniformly on the compacts of $\CC_+$. The function
$\bs p^v(z)$ is moreover the Stieltjes transform of a positive measure
$\bs\mu^{v}$.
By a passage to the limit, it holds that $| \imath y \bs p^v(-\imath y) + 1 |
\leq \varepsilon$ when $y > C_\varepsilon$.
This implies that $-\imath y \bs p^v(-\imath y) \to 1$ as $y\to\infty$,
in other words, $\bs \mu^{v}$ is a probability measure. We have shown that
the sequence of probability measures $\bs\mu_T$ is weakly compact. Therefore
it is tight.
\end{proof}

From the tightness of the sequence $\bs\pi_T$ and the
convergence~\eqref{cvg-st}, we get that any converging subsequence of
Stieltjes transforms of the $\mu_T$ has the Stieltjes transform of a
probability measure as a limit. Hence the tightness of the $\mu_T$ and the
second part of the statement of Theorem~\ref{approx}.

\subsection{Sketch of the proof of Theorem~\ref{capa}}
\label{prf-capa}
We remarked above that Equations~\eqref{eq-S}--\eqref{eq-tS} are similar to
their counterparts in~\cite[Th.~2.4]{hachem-loubaton-najim07}
(the latter are ``discrete frequency'' analogues of
Equations~\eqref{eq-S}--\eqref{eq-tS}). For this
reason, the proof of Theorem~\ref{capa} follows closely the proof
of~\cite[Th.~4.1]{hachem-loubaton-najim07}. We just provide here the main
steps of this proof.

The first observation is that the mutual information per receive antenna
$\int\log(1+\lambda) \mu_T(d\lambda)$ is related with $\bs m_T$ by the equation
\[
\int \log(1+\lambda) \, \mu_T(d\lambda) =
\int_1^\infty \Bigl( \frac 1t - \bs m_T(-t) \Bigr) \, dt
\]
which is obtained from the relation $d[\log(1+\lambda/t)]/dt =
-t^{-1} + (\lambda+t)^{-1}$. It also holds that
$\int \lambda \, \mu_T(d\lambda) < 4(\bs a^2 + \bs\sigma^2)$.
Indeed, getting back to the proof of Theorem~\ref{I-ids}, consider an event
where $\nu_T^n$ converges weakly to $\mu_T$ as $n\to\infty$, and where
$\limsup_n \int \lambda \nu_T^n (d\lambda) \leq 2(\bs a^2 + \bs\sigma^2)$.
Let $C > 0$ be a continuity point of $\mu_T$. Then for $n$ large enough,
\[
4(\bs a^2 + \bs\sigma^2) > \int (\lambda\wedge C) \,  \nu_T^n (d\lambda)
 \xrightarrow[n\to\infty]{}  \int (\lambda\wedge C) \,  \mu_T (d\lambda) .
\]
Taking a sequence of such continuity points that converges to infinity, we
obtain the result by the monotone convergence theorem.

By mimicking the proof of \cite[Lemma~C.1]{hachem-loubaton-najim07}, we can
also show that
$\int \lambda \, \bs \pi_T(d\lambda) < \bs a^2 + \bs\sigma^2$.
This shows that
\begin{align*}
& | t^{-1} - \bs p_T(-t) - (t^{-1} - \bs m_T(-t)) | \\
&\leq | t^{-1} - \bs p_T(-t) | + | (t^{-1} - \bs m_T(-t)) | \\
&=
\bigl| \int_0^\infty \Bigl( \frac 1t - \frac{1}{t+\lambda} \Bigr)
                                                 \bs\pi_T(d\lambda) \Bigr|
+
\bigl| \int_0^\infty \Bigl( \frac 1t - \frac{1}{t+\lambda} \Bigr)
                                                 \bs\mu_T(d\lambda) \Bigr| \\
&\leq t^{-2} \Bigl| \int \lambda \bs\pi_T(d\lambda) \Bigr|
 + t^{-2} \Bigl| \int \lambda \bs\mu_T(d\lambda) \Bigr| \leq
3(\bs a^2 + \bs\sigma^2) / t^2.
\end{align*}
Therefore,
\[
{\cal I}_T =
\int \log(1+\lambda) \, \mu_T(d\lambda) =
\int_1^\infty \Bigl( \frac 1t - \bs p_T(-t) \Bigr) \, dt
\]
is finite, and by the dominated convergence theorem,
$N^{-1} I_T(S; (Y,H)) - {\cal I}_T \to 0$ as $T\to\infty$. \\
In order to obtain the expression of ${\cal I}_T$ provided in the statement,
we need to find an antiderivative for $1/t - \bs p_T(-t)$.
This antiderivative can be obtained by a lengthy but straightforward
adaptation of the derivations of \cite[\S~C.2]{hachem-loubaton-najim07}.

\appendix

\section{Proof of Lemma~\ref{varQ}}
\label{anx-varQ}

Given an integer $M\in\NN$, let $H^M$ and $Q^M(z)$ be as in the statement of
Lemma~\ref{deriv}.
The variance of $Q_{p,q}^M(k,\ell)$ can be bounded by the PN inequality.
Specifically, we have by this inequality
$\var Q_{p,q}^M(k,\ell) \leq \chi_1 + \chi_2$ where
\begin{align*}
\chi_1 &= \sum_{i,j,u,v=-M}^M \sum_{n,t=0}^{N-1,T-1}
\EE[ X_{n,t}(i,j) \bar X_{n,t}(u,v)]
\EE\Bigl[ \frac{\partial Q_{p,q}^M(k,\ell)}{\partial X_{n,t}(i,j)}
\overline{\Bigl(\frac{\partial Q_{p,q}^M(k,\ell)}{\partial X_{n,t}(u,v)}\Bigr)}
\Bigr] \\
\chi_2 &= \sum_{i,j,u,v=-M}^M \sum_{n,t=0}^{N-1,T-1}
\EE[ X_{n,t}(i,j) \bar X_{n,t}(u,v)]
\EE\Bigl[
\overline{\Bigl(\frac{\partial Q_{p,q}^M(k,\ell)}{\partial \bar X_{n,t}(i,j)}
\Bigr)}
\frac{\partial Q_{p,q}^M(k,\ell)}{\partial \bar X_{n,t}(u,v)} \Bigr] .
\end{align*}
Let us deal with $\chi_1$. Since
\[
\frac{\partial Q_{p,q}^M(k,\ell)}{\partial X_{n,t}(i,j)} =
- \1_{|i|\leq M} \1_{|j|\leq M} \,
Q_{p,n}^M(k,i) [ H^{M*} Q^M]_{t,q}(j,\ell),
\]
we get
\begin{align*}
\chi_1 &= \frac 1T \sum_{i,j,u,v} \sum_{n,t}
\gamma_T(i-u) \phi_T(i-j)^2 \delta_{i-j , u-v}  \\
&
\ \ \ \ \ \ \ \
\EE\Bigl[ Q_{p,n}^M(k,i) [ H^{M*} Q^M]_{t,q}(j,\ell)
\overline{ Q_{p,n}^M(k,u) [ H^{M*} Q^M]_{t,q}(v,\ell)} \Bigr]  \\
&= \frac 1T \sum_{r, s, u} \sum_{n,t} \gamma_T(r) \phi_T(s)^2   \\
&
\ \ \ \ \ \ \ \
\EE\Bigl[ Q_{p,n}^M(k,r+u) [ H^{M*} Q^M]_{t,q}(r+u-s,\ell)
\overline{ Q_{p,n}^M(k,u) [ H^{M*} Q^M]_{t,q}(u-s,\ell)} \Bigr] .
\end{align*}
Using the inequality
$|\EE \sum_u X_u Y_u | \leq (\sum_u \EE|X_u|^2)^{1/2}
(\sum_u \EE|Y_u|^2)^{1/2}$, we get
\begin{align*}
\chi_1 &\leq \frac 1T \sum_{r, s} |\gamma_T(r)| \, \phi_T(s)^2  \sum_{u}
\sum_{n,t} \EE | Q_{p,n}^M(k,u) [ H^{M*} Q^M]_{t,q}(u-s,\ell) |^2 \\
&= \frac 1T \sum_{r, s} |\gamma_T(r)| \, \phi_T(s)^2  \sum_{u}
\EE \Bigl[ [ Q^M(k,u) Q^M(k,u)^* ]_{p,p}
[ [ H^{M*} Q^M](u-s,\ell)^* [ H^{M*} Q^M](u-s,\ell) ]_{q,q} \Bigr]  \\
&\leq \frac 1T \frac{\sigma_T^2}{y^2} \sum_r |\gamma_T(r)|
\sum_u \EE \Bigl[
[ [ H^{M*} Q^M](u,\ell)^* [ H^{M*} Q^M](u,\ell) ]_{q,q} \Bigr]  \\
&\leq \frac 1T \frac{\bs\sigma^2\bs g}{y^2}
\EE \Bigl[ [ Q^{M*} H^M H^{M*} Q^M]_{q,q}(\ell,\ell) \Bigr] \\
&\leq \frac{1}{T}\frac{\bs\sigma^2 \bs g (|z|+y)}{y^4} .
\end{align*}
The term $\chi_2$ can be treated similarly, using the second identity in the
statement of Lemma~\ref{deriv}. This leads to the inequality
\[
\var Q_{p,q}^M(k,\ell) \leq \frac{2}{T}\frac{\bs\sigma^2 \bs g (|z|+y)}{y^4} .
\]
Now, given any vector $a \in {\cal K}$, it is clear that
$H^M H^{M*} a \to HH^* a$ strongly as $M\to\infty$.
Since $\cal K$ is a core for $HH^*$, the operator $H^M H^{M*}$ converges to
$HH^*$ in the strong resolvent sense~\cite[Th.~VIII.25]{ree-sim-1}.
It results that $Q_{p,q}^M(k,\ell) \to Q_{p,q}(k,\ell)$ almost surely as
$M\to\infty$.
Since both these random variables are bounded by $1/y$, we get that
$\var Q_{p,q}^M(k,\ell) \to_{M\to\infty} \var Q_{p,q}(k,\ell)$, and the bound
on $\var Q_{p,q}(k,\ell)$ is established. \\

We now establish the bound over $\var \tr Q(k,\ell)$. Using the PN inequality
again, we get $\var\tr Q^M(k,\ell) \leq \chi_1 + \chi_2$ where
\[
\chi_1 = \sum_{i,j,u,v = -M}^M \sum_{n,t=0}^{N-1,T-1}
\EE[ X_{n,t}(i,j) \bar X_{n,t}(u,v)]
\EE\Bigl[ \frac{\partial \tr Q^M(k,\ell)}{\partial X_{n,t}(i,j)}
\overline{\Bigl(\frac{\partial \tr Q^M(k,\ell)}{\partial X_{n,t}(u,v)}\Bigr)}
\Bigr]
\]
and a similar formula for $\chi_2$. By the differentiation formula, we have
\[
\frac{\partial \tr Q^M(k,\ell)}{\partial X_{n,t}(i,j)} =
- \1_{|i|\leq M} \1_{|j|\leq M} \,
\Bigl[ [ H^{M*} Q^M](j,\ell) \ Q^M(k,i) \Bigr]_{t,n} .
\]
Therefore,
\begin{align*}
\chi_1 &= \frac 1T \sum_{i,j,u,v} \sum_{n,t}
\gamma_T(i-u) \phi_T(i-j)^2 \delta_{i-j,u-v} \\
&
\ \ \ \ \ \ \ \ \ \ \ \ \ \ \ \
\EE\Bigl[ [ [H^{M*}Q^M](j,\ell) Q^M(k,i) ]_{t,n}
\overline{ [ [H^{M*}Q^M](v,\ell) Q^M(k,u) ]_{t,n} } \Bigr]  \\
&= \frac 1T \sum_{r,s,u} \sum_{n,t} \gamma_T(r) \phi_T(s)^2 \\
&
\ \ \ \ \ \ \ \ \ \ \ \ \ \ \ \
\EE\Bigl[ [ [H^{M*}Q^M](u+r-s,\ell) Q^M(k,u+r) ]_{t,n}
\overline{ [ [H^{M*}Q^M](u-s,\ell) Q^M(k,u) ]_{t,n} } \Bigr]  \\
&= \frac 1T \sum_{r,s,u} \sum_{n,t} \gamma_T(r) \phi_T(s)^2 \\
&
\ \ \ \ \ \ \ \ \ \ \ \ \ \ \ \
\EE\Bigl[ [ H^{M*}Q^M D_\ell J_{k-\ell} Q^M]_{t,n}(u+r-s,u+r)
\overline{ [ H^{M*}Q^M D_\ell J_{k-\ell} Q^M ]_{t,n}(u-s,u) } \Bigr]
\end{align*}
where $D_\ell$ and $J_\ell$ are the operators on $l^2(\ZZ)$
\[
[J_\ell a](k) = a(k+\ell) \quad \text{and} \quad
[D_\ell a](k) = \delta_{k,\ell} a(k)
\]
where $a(k)$ (resp.~$[J_\ell a](k)$, $[D_\ell a](k)$) is the $k^{\text{th}}$
block with size $N$ of the vector $a$ (resp.~$J_\ell a$, $D_\ell a$) in the
canonical basis. Pursuing, we have
\begin{align*}
\chi_1 &\leq \frac 1T \sum_{r,s} |\gamma_T(r)| \, \phi_T(s)^2
\sum_u \sum_{n,t}
\EE\Bigl| [ H^{M*}Q^M D_\ell J_{k-\ell} Q^M]_{t,n}(u-s,u) \Bigr|^2 .
\end{align*}
Observe now that the operator
${\cal O} = H^{M*}Q^M D_\ell J_{k-\ell} Q^M Q^{M*} J_{k-\ell}^*
D_\ell Q^{M*} H^M$ is a finite rank operator such as
\[
\rank({\cal O}) \leq N \quad \text{and} \quad
\| {\cal O} \| \leq \frac{y+|z|}{y^4} .
\]
Therefore,
\[
\sum_u \sum_{n,t}
\Bigl| [ H^{M*}Q^M D_\ell J_{k-\ell} Q^M]_{t,n}(u-s,u) \Bigr|^2
\leq \tr {\cal O} \leq N \frac{y+|z|}{y^4}
\]
which shows that
\[
\chi_1 \leq \frac{\bs c \bs\sigma^2 \bs g (y+|z|)}{y^4}
\]
and we get the same bound for $\chi_2$. It remains to use the same argument
as in the first part of the proof to obtain the bound on $\var \tr Q(k,\ell)$.
\\

The derivations leading to the bound on $\var [H^* Q]_{p,q}(k,\ell)$ are
omitted, since they are similar to above.

\bibliographystyle{plain}


\def\cprime{$'$} \def\cdprime{$''$} \def\cprime{$'$} \def\cprime{$'$}
  \def\cprime{$'$} \def\cprime{$'$}

\end{document}